\documentclass[preprint,12pt,authoryear]{elsarticle}

\usepackage{amsmath}
\usepackage{amssymb}
\usepackage{bbm}

\newtheorem{lem}{Lemma}
\newtheorem{prop}{Proposition}
\newdefinition{rem}{Remark}
\newproof{proof}{Proof}

\newcommand{\intpos}[2]{\int_0^\infty{#1\rd{#2}}}
\newcommand{\e}{{\mathrm e}}
\newcommand{\pp}[1]{\left({#1}\right)}
\newcommand{\ind}{\mathbbm{1}}

\DeclareMathOperator{\bbE}{\mathbb{E}}
\DeclareMathOperator{\bbP}{\mathbb{P}}
\newcommand{\E}[1]{\mathbb{E}\!\left[#1\right]}
\newcommand{\Var}[1]{\mathrm{Var}\left({#1}\right)}
\newcommand{\rd}{{\mathrm d}}
\newcommand{\Laone}{\Lambda_1}
\newcommand{\dbmd}{d^\ast}
\newcommand{\prob}[1]{\mathbb{P}\left(#1\right)}
\newcommand{\eval}[2]{\left. #1 \right\vert_{#2}}
\newcommand{\abs}[1]{\left\lvert#1\right\rvert}

\journal{Mathematical Biosciences}

\begin{document}

\begin{frontmatter}



\title{Real-time growth rate for general stochastic SIR epidemics on unclustered networks}


\author[maths]{Lorenzo Pellis\corref{cor}}
\author[stats]{Simon Spencer}
\author[Man,maths]{Thomas House}

\cortext[cor]{Corresponding author: l.pellis@warwick.ac.uk}

\fntext[maths]{Warwick Infectious Disease Epidemiology Research centre (WIDER) and Warwick
Mathematics Institute, University of Warwick, Coventry, CV4 7AL, UK.}
\fntext[stats]{Warwick Infectious Disease Epidemiology Research centre (WIDER) and Department of Statistics, University of Warwick, Coventry, CV4 7AL, UK.}
\fntext[Man]{School of Mathematics, University of Manchester, Manchester, M13 9PL, UK.}

\begin{abstract}
	\noindent{}Networks have become an important tool for infectious disease epidemiology.
	Most previous theoretical studies of transmission network models have either considered
	simple Markovian dynamics at the individual level, or have focused on the
	invasion threshold and final outcome of the epidemic. Here, we provide
	a general theory for early real-time behaviour of epidemics on large configuration
	model networks (i.e.~static and locally unclustered), in particular focusing on 
	the computation of the Malthusian
	parameter that describes the early exponential epidemic growth. Analytical, numerical and Monte-Carlo methods under a wide variety of Markovian and non-Markovian 
	assumptions about the infectivity profile are presented. Numerous examples provide explicit quantification of the impact of the network structure on the temporal dynamics of the spread of infection and provide a benchmark for validating results of large scale simulations.
	\end{abstract}

\begin{keyword}
Epidemic \sep Malthusian parameter \sep Basic reproduction number \sep
Configuration model \sep Branching process


\MSC[2010] 05C82 \sep 60K20 \sep 92D30

\end{keyword}

\end{frontmatter}

\section{Introduction}
\label{sec:Introduction}

The field of infectious disease epidemiology has benefitted from the use of
networks both as conceptual tools and as a practical representation of
interaction between the agents involved in the spread of
infections~\citep{Bansal:2007,Danon:2011}.  From a theoretical perspective,
they have been successfully used to obtain important insight in the behaviour
of epidemics in idealised populations. Most analytical results, however, have
either been derived in the specific case of a Markovian SIR model involving
constant infection and recovery rates
\citep[e.g.][]{DiedJoMet98,Eames:2002,Volz:2008,Ball:2008,Miller:2011,Decreusefond:2012,Barbour:2013,Janson:2013}
or involve quantities that do not depend on the temporal details of the disease
dynamics \citep[e.g.][]{Newman:2002,Kenah:2007,Ball:2008,Ball:2009}. In this
paper we consider general non-Markovian SIR epidemic models and focus our
attention on arguably the most important piece of information concerning the
system's temporal dynamics: the \emph{Malthusian parameter}, or \emph{real-time growth rate}. This quantity corresponds to the rate of exponential growth in the
number of infectives observed in many models when an epidemic takes off in a
large population and the susceptible population is still large enough that its
reduction can be ignored. 

Realistic patterns of contact between people typically involve repeated
interactions with the same individual, and significant heterogeneity in the number
of contacts reported~\citep{Danon:2012}.  From an analytical point of view,
such a population structure is associated with three problems that need to be
addressed before the epidemic dynamics can be fully understood. These problems
are called \emph{repeated contacts}, \emph{infection interval contraction}, and
\emph{generational ordering}. We begin by illustrating these problems on a
model scenario.

%

\begin{figure}%
\includegraphics[width=\columnwidth]{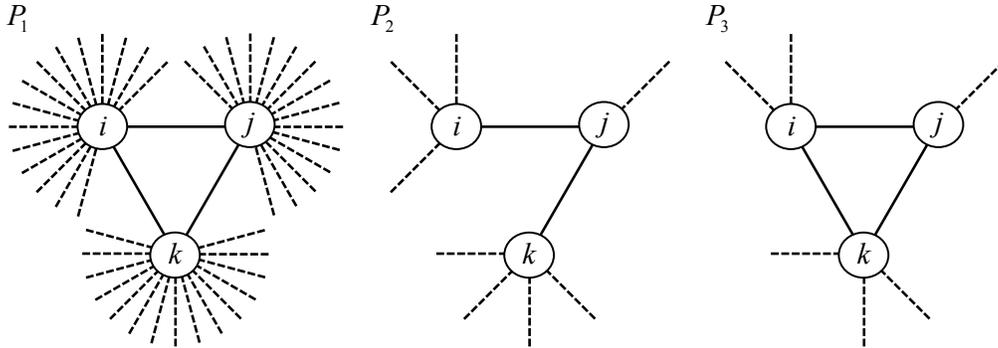}%
\caption{Schematic representation of three populations, $P_1, P_2$ and $P_3$. In $P_1$ 
mixing is homogeneous and hence all individuals are connected to each other. In
$P_2$ and $P_3$, $i$, $j$ and $k$ all have a small number of neighbours they can potentially infect, although in $P_2$ the network is tree-like, while in $P_3$ the network presents a triangle between $i, j$ and $k$.}%
\label{threepops}%
\end{figure}

Consider an infection spreading between individuals who are able to transmit when they enter in contact with each other. Throughout the paper we use the term
contact to mean an \emph{infectious contact}, i.e.~a contact that is `sufficiently intense' to result in an infection
whenever the individual that makes the contact is infectious and
the other is susceptible. A precise mathematical description of the model is given in Section \ref{sec:ModelDefinitions}, but in short we assume that, upon infection, individuals are attributed a (potentially time-varying) rate at which they make infectious contacts. When normalised, this gives the \emph{infectious contact interval distribution}, i.e.~the distribution of 
times between an individual becoming infected and the subsequent infectious contacts made by that individual (see also \citealp{Sve07,ScaliaTomba:2010,Kenah11}). 

When an infectious contact occurs, the infective chooses an individual to be the `destination' of the contact among a set of potential ones. Consider now a population $P_1$ where homogeneous mixing is
assumed, so that any pair of individuals can make infectious contacts, and compare it with
populations $P_2$ and $P_3$ where each individual can only make infectious contacts with individuals in a fixed set of neighbours (see Figure \ref{threepops}). In all cases we suppose that we are working with a large population of size $N\gg 1$, meaning that we will ignore effects that are, for example $O\pp{N^{-1}}$.  We represent populations $P_2$ and $P_3$ with a static (undirected) network where edges connect pairs of neighbouring individuals. 
In Section \ref{sec:ModelDefinitions} we regard the rate of making infectious contacts across each specified edge as our basic model ingredient. Note that population $P_1$ also admits a network representation, namely the complete network where each node is connected to all others; however, due to our assumption of a large population, the rate of making infectious contacts across each specific edge vanishes and the neighbouring relation becomes somewhat ephemeral.

We now focus our attention on a group of 3 individuals, labelled $i$, $j$ and $k$, where we assume that $i$ is infectious and the other two are susceptible, and we contrast the epidemic dynamics in $P_1$ with those in $P_2$ and $P_3$ (see Figure \ref{threepops}). 

First of all, whilst in $P_1$ individual $i$ can potentially infect
everybody in the population, in $P_2$ and $P_3$, no matter how infectious $i$ is, they cannot infect more than four other individuals, thus `wasting' part of their 
infectivity on \emph{repeated contacts} with the same individuals.

Secondly, in $P_1$ each infectious contact is made with an 
individual chosen at random and, because the population is large, it is unlikely 
for $i$ to contact the same individual more than once, so the average time at 
which transmission (to $j$, say) occurs is
equal to the average time at which $i$ makes an infectious contact with $j$. In
$P_2$ and $P_3$, where $i$ may try to infect $j$ multiple times, the time of transmission corresponds to the time of the \emph{first} infectious
contacts, which will occur on average earlier than the mean time at which a randomly selected infectious contact from $i$ to $j$ occurs.
We call this the problem of \emph{infection interval contraction}
. In choosing this terminology we have tried to avoid the very similar term `generation interval contraction', which is already adopted in \cite{KenLipRob08} to denote a slightly different phenomenon. Their approach based on survival analysis takes into account automatically what we here call the infection interval contraction. What they focus on instead is the fact that the time at which a susceptible is infected is the minimum of the times of all first infectious contacts from all potential infectors who are competing for infection of the same susceptible. The assumptions we make below 
avoid this problem as they imply that there is never more than one potential infector. 
In other words, the difference between infection interval contraction and generation interval contraction can be explained by the two different comparisons being made. In \citet{KenLipRob08}, different levels of infection prevalence are compared within a single epidemic model; on the other hand, here we are comparing a network-based epidemic (where repeated contacts occur) to a mass-action model (where each infectious contact leads to an infection). In particular, because we are interested in characterising the real-time growth rate, we assume we are in the early epidemic phase when the number of cases grows exponentially, i.e.~in our context the prevalence is assumed to be small.

Thirdly, once $i$ has infected $j$ in $P_2$ and $P_3$, then $j$ cannot infect $i$ and
`loses' a neighbour due to \emph{generational ordering}. Considering the next
generation of infection, if $j$ infects $k$ (an individual with only one
contact) then $k$ actually plays no further role in transmission events
due to this effect. In $P_1$, however, $j$'s infector is a negligible fraction
of its total neighbours and so this effect is insignificant.

What differentiates $P_3$ from $P_2$ is the presence of a triangle between $i, j$ and $k$. In the presence of short loops in the network (as it is the case, for example, if a small number of individuals all live in the same household), in addition to the previous three, two further effects become important. 

Fourth, if $i$ infects $j$ in $P_1$, the probability 
that $j$ infects $k$ before $i$ does is negligible. 
Therefore the event that $i$ infects $k$ is not 
affected by the epidemic in the neighbourhood of $i$ or 
what $i$ did before infecting $j$. The same occurs in 
$P_2$, because of the lack of a transitive link between 
$i$ and $k$. In $P_3$, instead, $j$ can infect $k$ 
before $i$ has the chance to, and so the number of 
susceptibles around $i$ can decrease because of 
infectious events not caused directly by $i$. We refer 
to this as the effect of \emph{local saturation of 
susceptibles}. 

Fifth, if $i$ infects $j$, $j$ infects $k$ and $i$ 
subsequently tries to infect $k$ (an event which does 
not result in an infection as $k$ is not susceptible 
any longer), we have two different ways of attributing 
cases to different generations of infection. Of course, 
assuming that $i$ is in generation 0, the natural 
choice would be to maintain the real-time description 
of who infects whom and place $j$ in generation 1 and 
$k$ in generation 2. However, this approach is 
analytically difficult to handle. The alternative is to 
consider all those that $i$ would have infected (both 
$j$ and $k$) and place them both in generation 1, in 
which case the transmission event from $j$ to $k$ is 
then overlooked \citep{Lud75,PelFerFra08}. We term this 
phenomenon \emph{overlapping generations}. In 
\citet{PelBalTra12} and \citet{BalPelTra14} the 
impact of this problem in defining and computing the 
basic reproduction number $R_0$ in models that involve 
small mixing groups (e.g.~households or workplaces) is 
carefully examined.

In generic social structures, these five effects often occur simultaneously. However, in this paper we consider scenarios in which they can be discriminated and progress can be made analytically, namely we focus on population with the structure of $P_2$, where the problems of local saturation of susceptibles and of overlapping generations described above need not be considered. For this reason, we make the strong assumption that that the proportion of possible transitive links is negligible
. This means that the network is locally tree-like. While some progress is possible for 
real-time growth rate calculations in the presence of many short loops in the
network (like in $P_3$), this often involves either restriction to compartmental dynamics or
approximate methods \citep{Fraser:2007,Ross:2009,Pellis:2010}. Analysis of
general stochastic dynamics of the kind we present here in the presence of
short loops would require a major conceptual advance not currently available, although see \citealp{BalPelTra14}, for a deeper exploration of these problems in the context of households models.


In addition to our first assumption of a locally tree-like network, we make two further assumptions. The second assumption is that the individuals are sufficiently `close' to each other
that the epidemic does indeed grow exponentially during its early phase rather
than more slowly (e.g.~quadratically, as would be expected on a two-dimensional lattice). More formally, this assumption requires that mean path lengths on the giant component are $O(\ln N)$.
Thirdly, we assume that the distribution of numbers of contacts is not too
heterogeneous (technically, that the second moment of the degree distribution
does not diverge) so that the Malthusian parameter does not diverge as the
system size becomes very large.

The three assumptions detailed above are fundamental to our approach, in the
sense that it is not clear how to analyse real-time behaviour of an epidemic on
a network mathematically if they do not hold. We also make other assumptions
that significantly simplify the analytical results obtained, but which can be
relaxed with a certain amount of elementary but potentially tedious algebra.
In particular, throughout we assume individuals are all identical to each
other, in the sense that there is no variation in susceptibility and in the
network model they only differ in terms of their degree. However, a key feature
of the present work is a careful treatment of an individual's infectiousness,
which is allowed to vary over time, according to some stochastic infectivity
profile. These infectivity profiles need not be the same for every individual,
but we assume that they are independent realisations from a specified
distribution that is the same for every individual. We also assume that there
are no degree-degree correlations, meaning that (given our other technical
assumptions) the configuration model can be used for the contact structure.

In setting up our framework, we took inspiration from \citet{DiedJoMet98}, where the authors focus on the computation of $R_0$ in the presence of repeated contacts under the same unclustered network approach discussed above. However, we here extend their work by adding the computation of the real-time growth rate and by moving from their deterministic framework to a stochastic one. Any results involving the real-time growth rate has its foundations on the so called \emph{Lotka-Euler} equation, originated in the field of demography and discussed, for example, in \citet{WalLip07} and references therein. 
In the simple case of a epidemic spreading in a large and homogeneously mixing population, the Lotka-Euler equation is derived as follows (see \citealp{DieHeeBri13}, p.~10 and 212). Assume the population has size $N$ and each infective, $t$ units after
having been infected, makes infectious contacts with other individuals at a rate $B(t)$. Then the
expected incidence $i(\tau)$ at absolute time $\tau$ satisfies the renewal equation 
\begin{equation}
\label{simplesystem}
i(\tau)=\frac{S(\tau)}{N}\int\limits_{0}^{\infty }{B(t)i(\tau- t )\:\text{d}t},
\end{equation}
where $S(\tau)$ is the number of susceptibles at time $\tau$. Early on, when almost the entire population is susceptible the incidence grows exponentially, $S(\tau)/N\approx 1$ and $B(t)$ in this linearised process describes the rate at which new cases are generated by a single infective at time $t$ after their infection. Substituting the Ansatz $i(\tau) = i_0 \e^{r \tau}$ into \eqref{simplesystem},
we deduce that the real-time growth rate $r$ must satisfy the Lotka-Euler
equation:
\begin{equation}
\label{LotkaEulerHomoMix}
\int\limits_{0}^{\infty }{B (t ) \e^{-r t}\:\text{d}t} = 1.
\end{equation}
Since the left-hand side of Equation \eqref{LotkaEulerHomoMix} is continuous in $r$ and has limits of 0
and $\infty$ (assuming
$B$ is not zero almost everywhere), a solution must exist by the intermediate value theorem. It is also strictly monotonic in $r$,
and so any solution must also be unique.

While we believe the question of real-time behaviour of epidemics on networks
is of inherent theoretical interest, the main aim of this work is
three-fold: (i)	to
make clear the key differences between the epidemic dynamics in homogeneous
mixing populations and on locally unstructured networks; (ii) to provide a comprehensive list of results
in analytically tractable cases, that can be used by modellers to validate
the outputs of complex simulations, for example individual-based ones; (iii) to provide the tools to assess the accuracy of approximations
for the real time growth rate that involve either ignoring the network structure altogether or neglecting part of the complexities to achieve significantly simpler analytical results, as done in \citet{Fraser:2007} and \citet{Pellis:2010}.


The rest of the paper is structured as follows. In Section \ref{sec:ModelDefinitions} we present our modelling approach, spelling out the assumptions behind the network construction and the epidemic spreading on it, and we define the fundamental quantities of interest in the network model and their counterparts in the limiting homogeneously mixing scenario. We then consider a list special cases in Section \ref{sec:SpecialCases}, for which we provide extensive analytical results. In Sections \ref{sec:OtherRandomTVIModels} and \ref{sec:NumericalMethod} we discuss more general analytical considerations and fully general Monte-Carlo methods for numerically computing the real-time growth rate. Finally, extensive numerical illustrations are presented in Section \ref{sec:NumericalResults} and final considerations in Section \ref{sec:conclusions}.


\section{Model definitions}
\label{sec:ModelDefinitions}

\subsection{Configuration network model}
\label{sec:NetworkConstruction}

We consider a static network representing a population of individuals and their
interactions. We assume the population size $N$ is large relative to the number
of infected individuals, which is in turn significantly larger than unity, as
we are interested in studying the asymptotic early spread of the infection.

To ensure that the network is locally unclustered, we assume the network is
constructed according to the so-called \emph{configuration model}
\citep{Molloy:1995} with the conventions most commonly used in epidemiology
\citep{Ball:2008,Ball:2010}. In this model, each individual $i$
($i=1,2,\dots,N$) is given a degree $d_i$ which is an independent realisation
of a non-negative integer-valued random variable $D$. The distribution of $D$
is called the \emph{degree distribution}. Individual $i$ is then allocated
$d_i$ stubs (half-edges). If the total number of stubs in the population is
odd, the process of attributing a degree to each individual is repeated. If the
number of stubs is even, stubs are selected at random without replacement and
paired.  The resulting network might have some short loops and self-edges, but
if a sequence of such networks is constructed and indexed with the population
size $N$ and we assume that $D$ has finite mean and variance, short loops and
self-edges appear with probability of order $O(N^{-1})$ and can be ignored as
$N\to\infty$ \citep{Durrett:2007}.


\subsection{The infectivity profile}
\label{sec:TheInfectivityProfile}

The function that describes the way an individual's infectiousness changes
through time is called their \emph{infectivity profile}. We make the
simplifying assumptions that individuals do not differ in terms of biological
susceptibility and that the infectivity profiles of different individuals are
independent and identically distributed according to some known distribution.
In addition, we assume that individuals can only transmit the infection across
edges, that each individual's infectious behaviour applies equally across all links with neighbours (in the precise way specified below) and that the infectious state of the individuals does not modify or
affect the network structure.

When infected, individual $i$ is attributed a realisation $\lambda _i(t)$ from
the infectivity profile distribution, where $t\geq 0$ is the time since the
infection of $i$. For each neighbour in the network, $i$ makes infectious
contacts at the points of an inhomogeneous Poisson process with rate
$\lambda_i(t)$. Thus, given $\lambda_i(t)$ for individual $i$, the probability
that $i$ makes infectious contact with a given neighbour in the interval
$(t,t+\Delta t)$ is $\lambda_i (t)\Delta t+o(\Delta t)$. If the individual
contacted by $i$ is susceptible, then they become infected. We assume that the
Poisson processes governing the infectious contacts between every pair of
individuals are conditionally independent given the infectivity profile of
individual $i$.


We denote by $\Lambda(t)$ ($t\geq 0$) the random infectivity profile, which is
the same for every individual and of which the $\lambda_i(t)$s are independent
realisations and we assume that $\bbE[\int_0^\infty {\Lambda(t) \text{d}t}] <
\infty$ and $\int_0^\infty \bbE[\Lambda(t)] \text{d}t < \infty$ and so by
Fubini's Theorem these two integrals must be equal. 

\subsection{The basic reproduction number $R_0$}
\label{sec:R0} 

We begin by calculating $R_0$, defined to be the expected number of secondary
infections caused by a typical infective in the early stages of the epidemic.
Let the random variable $A = \int_0^\infty {\Lambda(t)\text{d}t}$ represent the
total infectivity spread by an arbitrary infective across each of the links they have. Conditional on degree $D=d$ and
total infectivity $A=a$, the number of secondary cases that an infective will cause
during the early stages of the epidemic is distributed ${\rm Bin}(d-1,1-{\rm
e}^{-a})$. The number of trials $d-1$ takes into account the fact that the
infective must have acquired the infection from one of its neighbours, who is
therefore no longer susceptible (the generational ordering effect mentioned in the introduction). All other neighbours are almost surely
susceptible because in the early stages almost all of the population is
susceptible and short loops in the network appear with negligible probability. The probability of
transmission occurring across an edge between the infective and a susceptible
neighbour ($1-\e^{-a}$) takes into account that \emph{at
least} one infectious contact across a link is necessary, but all subsequent contacts are
ineffective. Note that the infection of each of the susceptible neighbours are
not independent events, but are conditionally independent given $A=a$.

A typical infective during the early stages of the epidemic has more
connections than the average individual in the network. The probability that
transmission across an edge reaches an individual of degree $d$ is equal to
\mbox{$d\times \mathbb{P}(D=d) / \mathbb{E}[D]$}, which we call the
\emph{degree-biased distribution}. Conditional on having total infectivity $a$,
this individual of degree $d$ then infects on average $(d-1) \pp{1-{\rm
e}^{-a}}$ new cases. Since an individual's degree $D$ and their infectivity $A$
are independent, removing the conditioning gives 
\begin{equation}
\label{R0Direct}
R_0 = \frac{\bbE[D(D-1)]}{\bbE[D]} \mathbb{E}\left[1-{\rm e}^{-A}\right].
\end{equation}
Hereafter we denote the mean degree of the node reached by a randomly selected edge by
\begin{equation}
\label{degreefactor}
\dbmd = \frac{\E{D(D-1)}}{\E{D}}.
\end{equation}

\subsection{The real-time growth rate}
\label{sec:TheRealTimeGrowthRate}

In this paper, our primary interest is studying the real-time growth rate -- the rate of the exponential growth in the number of infectives during the early phase of an emerging epidemic -- which is denoted by $r$. Unlike
$R_0$, the real-time growth rate depends on the full infectivity profile
$\Lambda(t)$ and not just its integral $A$. Of all the infectious contacts
towards a susceptible neighbour, only the first contact results in an
infection. The rate at which infections occur is therefore equal to the rate at
which contacts are made multiplied by the probability that no contacts have
occurred previously. Let $\Laone(t)$ denote the rate at which \emph{first}
contacts occur, and hence
\begin{equation}
\label{DefLaone}
\Laone(t) = \Lambda (t) \e^{-\int_0^t{\Lambda(s)\rd s}}.
\end{equation}
Note that $\intpos{\Laone(t)}{t} = 1-\e^{-A}$. We can also describe the distribution of times since infection at which infectious contacts are made, which has probability density function (pdf) $w(t)=\E{\Lambda(t)}/\E{A}$ and the distribution of times since infection at which infections are made, which has pdf $w_1(t)=\E{\Laone(t)}/\E{1-\e^{-A}}$. We refer to the former as the \emph{infectious contact interval distribution} and the latter as the \emph{generation interval distribution}. Let the random variables $W$ and $W_1$ denote draws from these two distributions, which are both implicitly conditional on the fact that at least one infectious contact is made across the edge. The mean generation interval, often referred to in the literature as \emph{generation time}, of the epidemic on the network is therefore
\begin{equation}
\label{DefTgNet}
T_g = \E{W_1}.
\end{equation}

Using again the degree-biased distribution to compute the degree of a randomly
selected new infective, the \emph{average infection intensity} of a typical
infective is therefore
\begin{equation}
\label{BetaDirect}
\beta(t) = \dbmd \E{\Laone(t)}.
\end{equation}

As for the derivation of Equation \eqref{LotkaEulerHomoMix}, given the rate $\beta(t)$ at which a typical infective generates new infections $t$ units of time after their infection, the
expected incidence $i(\tau)$ at absolute time $\tau$ satisfies the
renewal equation 
\begin{equation}
\label{RenewalNet}
i(\tau)=\int\limits_{0}^{\infty }{\beta (t)i(\tau- t )\:\text{d}t}.
\end{equation}
Substituting the Ansatz $i(\tau) = e^{r\tau}$, the real-time growth rate $r$ must satisfy
\begin{equation}
\label{LotkaEuler}
\int\limits_{0}^{\infty }{\beta (t ) \e^{-r t}\:\text{d}t} = 1,
\end{equation}
which has the same form of Equation \eqref{LotkaEulerHomoMix} because all the complexities due to the network structure are absorbed in $\beta(t)$. Like for Equation \eqref{LotkaEulerHomoMix}, the solution exists (provided
$\beta$ is not zero almost everywhere) and is unique.
To our knowledge, the most
comprehensive derivation of this equation from an underlying stochastic process on a network
is given by \citet{Barbour:2013}, and in a deterministic framework by
\citet{OD14}. Note also
that $\int_{0}^{\infty }{\beta (t ) \:\mathrm{d}t} = R_0$, so Equation \protect\eqref{LotkaEuler} can be also formulated as
\begin{equation}
\int\limits_0^\infty{w_1 (t ) \e^{-r t}\:\text{d}t} = \frac{1}{R_0},
\label{rR0rel}
\end{equation}
from which it is clear that, for
$R_0<1$, $R_0=1$ and $R_0>1$, the real-time growth rate $r$ is strictly
negative, zero, and strictly positive, respectively.

One does not observe the real-time growth rate in practice unless $R_0>1$ and
the epidemic takes off. However, if $R_0<1$, the rate $r$ of exponential decay can still be observed in theory if it is
assumed that the number of initial infectives is large enough to avoid
stochastic fluctuations, but still small relative to the total population.

\subsection{Homogeneously mixing model}
\label{sec:HomogeneouslyMixingModel}

To understand the effects that the configuration network structure has on the
dynamics of the infection, it is useful to compare the configuration network
model to one in which individuals mix homogeneously. In this model an infected
individual makes contacts with every individual in the population with equal
probability and so this model can be thought of as an epidemic on a complete
graph. Clearly the complete graph does not have the locally tree-like structure
of the configuration network; however, since each individual in a population of
size $N$ has $N-1$ neighbours, the probability of an individual attempting to
contact the individual that infected them is $O(N^{-1})$ and can be neglected
when $N$ is large. Hence in its early stages, the epidemics again has a
branching structure and all five problems spelt in the introduction disappear as $N$ tends to infinity. For more details
see  \citet{Barbour:2013}.


More formally an infected individual is allocated an infectivity profile 
$\Lambda^h(t)$ and makes infectious contacts to each of their $N-1$ neighbours 
at the points of an inhomogeneous Poisson process with rate $\Lambda^h(t)$, 
where $t$ is the time since they were infected and the superscript $h$ refers 
to the homogeneously mixing population. 
The expected total rate at which an infective makes infectious contacts is therefore $\beta^h(t)=(N-1)\E{\Lambda^h(t)}$. 
As we let the population size $N$ tend to infinity, we reduce $\Lambda^h(t)$ such that $\beta^h(t)$ remains fixed.


In the limit of an infinitely large homogeneously mixing population, the probability that an infective contacts the same individual more than once is zero. Therefore in the early stages of the epidemic, when the population is almost entirely susceptible, every infectious contact results in an infection. The expected rate at which infections are made is therefore $\beta^h(t)$ -- the same as the rate at which infectious contacts are made. This is the key difference between the homogeneously mixing model and the network model, as in the network model infections occur proportionally to the rate of the expected \emph{first} contact. 

Note that the basic reproduction number for the homogeneously mixing model is
given by $R_0^h = \int_0^\infty{\beta^h(t)}\rd t$. Similarly, let $w_1^h(t)$
denote the \emph{generation interval distribution}, which represents the
distribution of times relative to one's infection at which infections are made
in a large homogeneously mixing population. Because every contact results in an infection, in a large and homogeneously mixing population the contact interval distribution and the generation interval distribution are the same. Analogously to the configuration
network model $w_1^h(t) = \beta^h(t) / R_0^h$. Note that $w_1^h$ can also
be interpreted as the pdf of a random variable, $W_1^h$ say, representing the
time of an infection from a typical infective, conditional on at least one
occurring. Hence, the \emph{generation time} $T_g^h = \E{W_1^h}$ represents the
mean time interval between the infection of a case and one randomly selected
secondary infection, provided that at least one occurred. Alternatively (see
\citealp{Sve07}) $T_g^h$ can equivalently be defined as the time interval
between the infection of a case and the infection of that case's infector.
Finally let the real-time growth rate for the homogeneously mixing model be
denoted by $r^h$, which is the unique solution to the Lotka Euler
equation
\begin{equation}\label{LotkaEulerHomo}
\int\limits_0^\infty\beta^h(t)\e^{-r^ht}\:\mathrm{d}t=1.
\end{equation}

\subsection{Approximation to the real-time growth rate}
\label{sec:Approximater}
For some infectivity profiles the function $\beta(t)=R_0w_1(t)$ may be difficult to compute, preventing the calculation of $r$  for the configuration network model. A possible approximation $\beta_{\text{app}}$ of $\beta$ involves replacing the generation interval distribution $w_1(t)$ with the infectious contact interval distribution $w(t)$. Hence,
\begin{equation}
\label{DefBetaApp}
\beta_{\text{app}}(t) = R_0 w(t).
\end{equation}

This approximation treats correctly the overall infectivity and the probability of infections occurring, but approximates
the times at which the real infections are made with the times at which
infectious contacts are made. The result is an approximation of the real-time
growth rate $r$ on a network, given by the unique value $r_{\mathrm{app}}$ that
satisfies $\int_0^\infty \beta_{\mathrm{app}}(t) {\rm
e}^{-r_{\mathrm{app}}t}{\rm d}t = 1$. Some authors \citep{Fraser:2007,Pellis:2010} have
considered a similar sort of approximation in the slightly more complex case of
a households model, where the additional problems of local saturation of
susceptibles and overlapping generations cannot be neglected (see also
\citealp{BalPelTra14}, for a careful analysis of this sort of approximation).

\subsection{Comparisons between models}
\label{sec:Comparisons}

In order to make meaningful comparisons between the configuration network and
the homogeneously mixing models it is necessary to match certain quantities.
In particular, we first need a matching condition involving the total infectivity (integrated over time) to establish the probabilities that events happen (irrespective of when). Second, we need a way to match the timescale dimension of the infectivity profile in the different models, to guarantee a meaningful comparison of the temporal dynamics.

The expected total infection pressure exerted by an individual in the network model is
$\dbmd\E{A}$, whilst in the homogeneous mixing model it is $R_0^h$. We take as our first matching condition that the overall infection pressure from each infective is
the same in each model, i.e.~$\dbmd\E{A} = R_0^h$. Another alternative matching criteria could have been assuming that $R_0=R_0^h$. The reason why we chose the former setting is that
$R_0=\dbmd\E{1-\e^{-A}}$ already takes into account some of the
 aspects of repeated contacts due to the network structure (through the probability of at least one infectious contact across an edge $1-\e^{-A}$) which is among the effects we wish to investigate. Furthermore, note that if $A$ is random, it is arguably more intuitive to keep fixed the mean of $A$ than the mean of the `distorted' random variable $1-\e^{-A}$. However, note that if we increase the connectedness of
the network model in such a way that $R_0$ is kept constant then we must have
that $\E{A}\to 0$, which implies that $\E{1-\e^{-A}}\rightarrow 0$ and therefore these two possible comparisons
are asymptotically equal.

The matching on the temporal dimension between the homogeneously mixing model and the network model is achieved by assuming that $\Lambda^h(t)\propto\Lambda(t)$ for all $t$, almost surely.

Within either model, care also needs to be taken when comparing the impact of
various infectivity profiles. Depending on the context, we might keep fixed $R_0^h$ or $R_0$ (i.e.~$\E{A}$ or $\E{1-\e^{-A}}$) when varying the infectivity 
profile. Finally, in terms of temporal scale, the choice of what to keep fixed in general depends on the details of of the infectivity profile formulation, so we specify clearly how the comparison is made in each numerical example.

\section{Special cases}
\label{sec:SpecialCases}

\begin{figure}[tb]
\includegraphics[width = 0.95\textwidth]{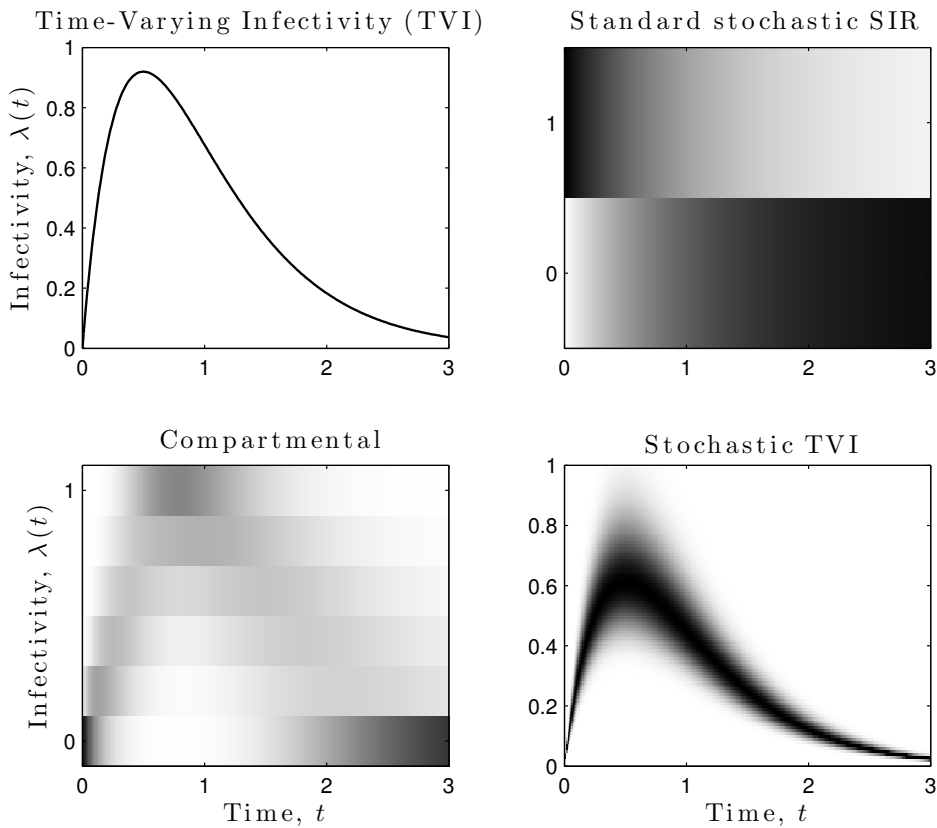}
\caption{Illustration of the four different model classes considered; a darker shade denotes higher probability. }
\label{modelfig}
\end{figure}

We now consider some specific infectivity profile distributions where more
explicit calculations of the real time growth rate can be made. Some of these special cases have already been considered in other studies (see e.g.~\citealp{WalLip07}), although predominantly in large homogeneously mixing populations and not in such generality. It is
convenient to categorise these examples into a few broad classes, as shown in
Figure~\ref{modelfig}. In the first class, referred to as the
\emph{time-varying-infectivity} class, the infectivity profile
$\Lambda(t)$ is non-random and therefore is the same for all individuals. The
second class is referred to as the \emph{standard stochastic SIR model} class,
and assumes constant infectivity during an infectious period with a random but
specified length. We then consider \emph{compartmental} models where the
infectious period is obtained from a phase-type distribution and infectivity
can depend on the compartment and we conclude with \emph{general time
varying-infectivity} models. 

\subsection{Non-random time-varying-infectivity models}
\label{sec:NonRandomTVIModels}

In the \emph{time-varying-infectivity} (TVI) class (see \citealp{Pellis:2010}),
the infectivity profile $\Lambda(t)$ is non-random and is therefore the same
for all infected individuals. Hence the total infectivity across a link $A$ is constant, as well as
the probability of transmission across a link, $1-\e^{-A}$. Therefore, $R_0 = \dbmd
(1-\e^{-A})$.

\subsubsection{Gamma-shaped TVI}
\label{sec:GammaShapedNonRandomInfectivity}

Assume that the infectivity profile is proportional to the pdf of a
$\Gamma$-distributed random variable with shape parameter $\alpha$ and scale
parameter $\gamma$. Let $f_{\Gamma}(t) = \gamma^\alpha t^{\alpha-1} \e^{-\gamma
t} / \Gamma(\alpha), t>0,$ denote the pdf of a Gamma random variable with these
parameters and let $F_{\Gamma}(t)$ denote the corresponding cumulative density function (cdf). Thus, $\Lambda
(t)=A f_{\Gamma}(t)$ in the network model and therefore the matched
homogeneously mixing model (see \ref{sec:Comparisons}) has
$\beta^h(t)=R_0^hf_{\Gamma}(t)$.

In the homogeneously mixing model there is an explicit formula for the
real-time growth rate. Here, $w^h(t) = f_{\Gamma}(t)$ and $T_g^h = \alpha
/ \gamma$. The Lotka-Euler equation \eqref{LotkaEulerHomo} becomes
$$R_0^h\int\limits_0^\infty f_{\Gamma}(t)\e^{-r^ht}\:\mathrm{d}t=1,$$
where the left-hand side is proportional to the moment generating function (mgf) of the Gamma
distribution evaluated at $r^h$, i.e.~$\pp{\frac{\gamma}{\gamma+r^h}}^\alpha$
. Hence, $r^h = \gamma \pp{
\pp{R_0^h}^{\frac{1}{\alpha}} - 1 }$.

On the network, we have 
\[ \Laone(t) = A f_{\Gamma}(t) \e^{-A F_{\Gamma}(t)} \]
and $w_1(t) = \Laone(t)/(1-\e^{-A})$. The Laplace transform of $w_1(t)$ appears to be intractable, so quadrature methods are then needed to compute $T_g$ and $r$ numerically\footnote{In terms numerical stability, when calculating the Laplace transform of $\Laone(t)$ for small values of $A$ (i.e.~when $r$ becomes negative and hence $\e^{-rt}$ explodes for large $t$), we recommend using the fact that $\Laone(t)\e^{-zt} = A \frac{\gamma}{z+\gamma} g(t)\e^{-AF_\Gamma(t)}$, where $g$ is the pdf of a Gamma distribution with parameters $\alpha$ and $z+\gamma$.}.

\subsubsection{Reed-Frost model}
\label{sec:ReedFrostModel}

In the \emph{Reed-Frost model} (see \citealp{AndBri2000}), each infected
individual experiences a latent period of the same fixed duration $T$ after
which they spread all their infectivity instantaneously. The Reed-Frost model
forms the limit when the variance in the times of infectious contacts
tends to zero.


There is only one point in time at which and an infectious contact and therefore an infection can occur, and so $T_g^h = T_g = T$, and hence $w_1^h(t) = w_1(t) = \delta_{T}(t)$, where $\delta_{\tau}(t)$ is the Dirac's delta function centred at $\tau$. Direct solution of the Lotka-Euler equation leads to $r^h = \ln\pp{R_0^h} / T_g^h$ and $r = \ln\pp{R_0} / T_g$.

\subsubsection{Heaviside TVI}
\label{sec:HeavisideFunction}

In this case the infectivity profile is some constant $\lambda$ between times
$u$ and $v$ and zero elsewhere. Hence, $\Lambda (t)= \lambda \ind_{\{ u\le
t<v\} }$ for the network model and $\beta^h(t)=R_0^h \ind_{\{ u\le t<v\} }$ for
the homogeneously mixing model, where $0\leq u<v$ and $\ind$ denotes the
indicator function.


For the homogeneously mixing model, $w_1^h(t) = \ind_{\{ u\le t<v\} } / (v-u)$, $T_g^h = (u+v)/2$ and $r^h$ satisfies the transcendental equation 
\[R_0^h (\e^{r^h u}-\e^{r^h v}) = r^h.\]

On the network, $\Laone(t)=\lambda {\rm e}^{-\lambda (t-u)}\ind_{\{u\leq
t<v\}}$ and hence
$$T_g = \frac{1}{\lambda}+\frac{u-v\e^{-\lambda(v-u)}}{1-\e^{-\lambda(v-u)}}.$$
The real-time growth rate $r$ satisfies the transcendental equation
\begin{equation} 
\label{rforHeaviside}
 \dbmd \frac{\lambda }{\lambda +r} \left(\e^{-ru} -\e^{-rv-\lambda (v-u)} \right) = 1.
 \end{equation} 

The parameter $u$ represents the length of the latent period. When $u=0$ this
model becomes the SIR model with constant infectious period, which is treated
in Section \ref{sec:ConstantInfectiousPeriod}.

\subsection{Stochastic SIR model}
\label{sec:StandardSIRModel}

In the \emph{stochastic SIR model} (hereafter denoted as sSIR), each infected
individual remains infectious for a period of duration $I$, where $I$ is a
non-negative random variable having arbitrary but specified distribution.
During the infectious period, infectious contacts towards each neighbour occur
at the points of a homogeneous Poisson process with rate $\lambda$. Therefore,
$\Lambda(t)=\lambda\ind_{\{0\le t < I\}}$, $A = \lambda I$ and $\E{\Lambda(t)}
= \lambda \pp{1-F_I(t)}$, where $F_I$ denotes the cdf of $I$.

For the homogeneously mixing model $\beta^h(t)=\frac{R_0^h}{\bbE[I]}\ind_{\{0\leq t<I\}}$
and $w_1^h(t) = \pp{1- F_I(t)} / \E{I}$.
Integrating by parts, we obtain
\begin{equation}
\label{TghsSIR}
T_g^h = \frac{\E{I^2}}{2\E{I}}.
\end{equation}
Similarly, the Lotka-Euler equation \eqref{LotkaEulerHomo} becomes
\begin{equation}
\label{rhforsSIR}
R_0^h \frac{1 - M_I(r^h)}{r^h \E{I}}=1,
\end{equation}
where $M_I(z)=\E{\e^{-zI}}$ is the moment generating function of $I$.

For the network model we calculate $T_g$ and hence $r$ as follows. Recall that the random variable $W_1$ represents the time at which infection crosses an edge, given that this infection does occur, and that $W_1$ has pdf $w_1(t)$. Let $U\sim {\rm Exp}(\lambda)$ be the time to the first event in a homogeneous Poisson process with rate $\lambda $, which represents a viable infectious contact if and only if it occurs before the recovery time of the infective. Hence, $W_1=U|(U<I)$. 

From the definition of conditional expectation,
\begin{align*}
M_{W_1} (z) & = \E{{\rm e}^{-zW_1} }\\  & =  \frac{\E{{\rm e}^{-zU} \ind_{\{ U<I\} } }}{{\rm {\mathbb P}}\left(U<I\right)} .
\end{align*}
Note also that
\begin{align*}
   \E{ {{\text{e}}^{-zU}}{{\ind}_{\{U<I\}}} } &=\E{\int\limits_{0}^{I}{\lambda {{\text{e}}^{-\lambda s}}{{\text{e}}^{-zs}}\;\text{d}s} } \\ 
 & =\E{\frac{\lambda }{\lambda +z}\left( 1-{{\text{e}}^{-(\lambda +z)I}} \right)} \\ 
 & =\frac{\lambda }{\lambda +z}\left( 1-M_I (\lambda +z) \right) .
\end{align*}

Letting $z\to 0$ in the calculations above we also obtain ${\rm {\mathbb
P}}\left(U<I\right) = 1-M_I (\lambda )$, and hence,
\begin{equation}
\label{FranksMGFofY}
M_{W_1} (z)=\frac{\lambda }{\lambda +z} \left(\frac{1-M_I (\lambda +z)}{1-M_I (\lambda )} \right).
\end{equation}
Equation \eqref{FranksMGFofY} has been previously derived by Prof.~F.~Ball and
independently by Prof.~O.~Diekmann (personal communications). As a note of
caution, $M_{W_1}(z)$ as defined in \eqref{FranksMGFofY} is undefined for $z =
-\lambda$, but can be extended by continuity to give $M_{W_1}(-\lambda) =
\lambda\E{I} / \pp{1-M_I(\lambda)}$.

The generation time $T_g = \E{W_1}$ can now be obtained via 
\begin{equation}
\label{TgsSIR}
T_g = \frac{1}{\lambda} + \frac{\eval{\frac{\rd M_I(z)}{\rd z}}{z = \lambda}}{1-M_I(\lambda)}.
\end{equation}
Note that the second term in \eqref{TgsSIR} is negative, and so $T_g <
1/\lambda$. The real-time growth rate $r$ is the unique solution of the
Lotka-Euler equation, which in this notation becomes $R_0 M_{W_1}(r) = 1$.  

\subsubsection{Markovian SIR model}
\label{sec:MarkovianSIRModel}

The notable special case of the sSIR model for which $I$ is exponentially
distributed is usually referred to as the \emph{Markovian SIR model} and is
typically the most analytically tractable. Let $I\sim {\rm Exp}(\gamma )$, so
that $\gamma$ represents the recovery rate. Then $\E{I} = 1/\gamma$ and $M_I(z)
= \gamma / (\gamma+z)$. 

For the homogeneously mixing model, $w_1^h(t)= \gamma {\rm e}^{-\gamma t} $, so that the random variable $W_1^h$, representing the time to an infection given that one occurs,  is Exp($\gamma$). Hence $T_g^h = 1/\gamma$, and 

\begin{equation}
\label{rhMark}
r^h =\frac{R_0^h-1}{T_g^h}.
\end{equation}

On the network, we have
\[\Laone(t)=\left\{\begin{array}{cc} {\lambda {\rm e}^{-\lambda t} } & {0\le t<I} \\ {0} & {{\rm otherwise}} \end{array}\right. \] 

and therefore $\E{\Laone(t)}=\lambda \e^{-(\lambda +\gamma )t}$. The probability of infection across an edge is $\E{1-\e^{-A}}=\lambda /(\lambda +\gamma )$, which can also be derived by a competing hazards argument, and $w_1(t)=(\lambda +\gamma ){\rm e}^{-(\lambda +\gamma )t} $. Therefore, $W_1\sim \mathrm{Exp}(\lambda + \gamma)$, whence $T_g= 1/(\lambda+\gamma)$ and $\beta(t) = R_0 \pp{\lambda+\gamma} \e^{-\pp{\lambda+\gamma}t}$. The real-time growth rate $r$ is therefore

\begin{equation}
\label{rforMark}
r = \frac{R_0 - 1}{T_g}.
\end{equation}
Note the strong similarity between Equations \eqref{rhMark} and
\eqref{rforMark}. If the homogeneously mixing model is matched to the network
model as described in Section \ref{sec:Comparisons} (which in this context means keeping both $\lambda$ and $\gamma$ constant), then we find
$r^h=\lambda\dbmd-\gamma$, whilst $r=\lambda(\dbmd-1)-\gamma$. The term $(\dbmd
- 1)$ in $r$ instead of $\dbmd$ in $r^h$ accounts exactly for the amount of
infectivity wasted in repeated contacts towards already infected neighbours.

\subsubsection{Gamma-distributed infectious period}
\label{sec:GammaDistributedInfectiousPeriod}
Assume that $I$ follows a gamma distribution with shape parameter $\alpha$ and scale parameter $\gamma$, i.e.~$I\sim \Gamma(\alpha,\gamma)$. In this case $\E{I} = \alpha / \gamma$ and $\E{I^2} = \alpha\pp{1 + \alpha}/\gamma^2$ and hence from \eqref{TghsSIR}
\begin{equation}
\label{TghsSIRGamma}
T_g^h = \frac{1+\alpha}{2\gamma},
\end{equation}
(see also \citealp{Sve07}). The Lotka-Euler equation \eqref{rhforsSIR} can be solved numerically to find $r^h$.

For the network model, we obtain from \eqref{TgsSIR} that
\begin{equation}
\label{TgsSIRGamma}
T_g = \frac{1}{\lambda} - \frac{\frac{\alpha}{\gamma+\lambda}\pp{\frac{\gamma}{\gamma+\lambda}}^\alpha}{1-\pp{\frac{\gamma}{\gamma+\lambda}}^\alpha} 
\end{equation}
and the real-time growth rate $r$ can be found numerically as the solution of $R_0 M_{W_1}(r) = 1$ using \eqref{FranksMGFofY}.

\subsubsection{Constant infectious period}
\label{sec:ConstantInfectiousPeriod}
The constant infectious period, in which $I=\iota$ almost surely, is the limiting case of a Gamma distributed infectious period where $\alpha \to \infty$ while $\E{I}=\alpha/\gamma$ is kept constant. This also corresponds to a special case of the the non-random Heaviside infectivity profile discussed in Section \ref{sec:HeavisideFunction}, when $u=0$ and $v=\iota$.
In this case, from \eqref{TghsSIR}, $T_g^h = \frac{\iota}{2}$ and, from \eqref{rhforsSIR}, $r^h$ satisfies
\begin{equation}
\label{rhforConst}
R_0^h\frac{1-\e^{-r^h\iota}}{r^h\iota} = 1.
\end{equation}

On the network, from \eqref{TgsSIR},
\begin{equation}
\label{TgforConst}
T_g = \frac{1}{\lambda} - \frac{\iota\e^{-\lambda\iota}}{1-\e^{-\lambda\iota}}
\end{equation}
and, from \eqref{FranksMGFofY}, $r$ is the unique solution of
\begin{equation}
\label{rforConst}
R_0\frac{\lambda}{\lambda+z}\frac{1-\e^{-(\lambda+z)\iota}}{1-\e^{-\lambda\iota}}=1.
\end{equation}

\subsection{Compartmental models}
\label{sec:compartmentalmodels}

We now consider the case where the disease progression in an individual is
modelled by progression through membership of disjoint compartments, connected in general to form any combination of paths in series and parallel, with
different infectivities in each of them. Such a Markov reward process (see e.g.~\citealp{asmussen2003}), also sometimes referred to as `linear chain trickery' (see \citealp{Bre+12}, p.~105 and 106, or
\citealp{OD14}), allows for a flexible modelling approach in which quantities of
interest can be calculated using relatively simple formulae.

The starting point for the compartmental approach is that there is a finite
number of disease states, labelled by integer indices $a,b,\ldots$ between 1
and $m$, each with its own rate of making infectious contact with other individuals. Let an individual have a random variable
$X(t)$ for its disease state at time $t$, which takes the value $S$ if the
individual is susceptible, $R$ if the individual is recovered, or $I_{a}$ if
the individual is in the $a$-th disease state.  We suppose that the probability
of a newly infected individual starting in state $I_{a}$ is $\nu_a$, that the
rate of going from $I_{a}$ to $I_{b}$ is $\sigma_{a,b}$, and that the rate of going
from $I_{a}$ to $R$ is $\mu_a$. This general framework is shown in
Figure~\ref{phasefig}. Note that infected individuals can potentially start or finish their infectious life from any state $I_a$ (for example, newly infected cases can develop either severe or mild symptoms with some fixed probabilities, and recover at different rates; or distinction between compartments can be only fictitious, with the aim of achieving sojourn times with distributions other than exponential -- see Section \ref{sec:ExamplesComp}) and that infectious states may be visited any number of times. 
We assume that all rates are non-negative and we exclude pathological cases, by assuming finite infectivity in each state, no infectivity in state $R$ and that individuals cannot be infectious forever, i.e.~infectives recover with probability 1.
For notational convenience, let $\nu = \pp{\nu_1,\nu_2,\dots,\nu_m}$ be a row vector, $^\top$ denote vector transposition and $\sigma_a := \sum_b \sigma_{a,b}$ denote the total rate
at which an individual in infectious state $a$ moves to other infectious
states.

\begin{figure}[tb]
\centering
\includegraphics[width = 0.7\textwidth]{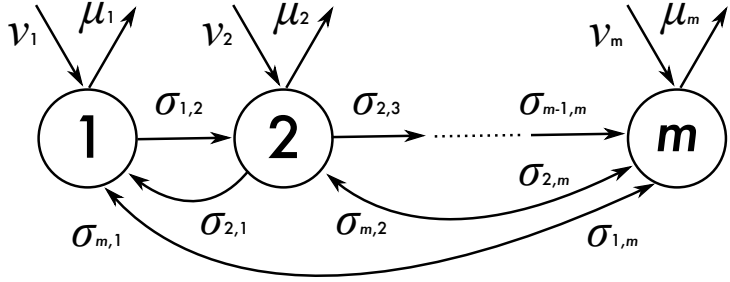}\\[1cm]
\caption{Illustration of a phase-type distribution.}
\label{phasefig}
\end{figure}

\subsubsection{Network mixing}
Let $\lambda_a$ be the rate at which an infective makes infectious contacts across an edge while in state $a$. From Sections \ref{sec:R0} and \ref{sec:TheRealTimeGrowthRate} we know that, in order to calculate all quantities of interest, we need to describe the average infection intensity $\beta(t)$. From Equation \eqref{BetaDirect}, $\beta(t) = \dbmd \E{\Laone(t)}$. This simple factorisation, which depends on our assumption that the network structure is unaffected by the infection process, implies that it is sufficient to study the transmission process across a single edge, as described by $\Laone(t)$. To characterise $\Laone(t)$ we consider the continuous-time Markov chain defined by events and rates
\begin{equation}\begin{aligned}
I_{a} & \rightarrow I_{b} & \text{ at rate } & \sigma_{a,b} \text{ ,}\\
I_{a} & \rightarrow R & \text{ at rate } & \mu_{a} \text{ ,} \\
I_{a} & \rightarrow J & \text{ at rate } & \lambda_a \text{ ,} 
\label{augtrans}
\end{aligned}\end{equation}
where state $J$ represents the situation where an infectious contact has been made and the edge cannot transmit the infection anymore. We decompose the state space $\mathcal{S}$ of this Markov chain in $\mathcal{S} =
\mathcal{A}\cup\mathcal{C}$, where $\mathcal{C}=\{I_a\}_{a=1,\dots,m}$ are the transient
states of the system and $\mathcal{A} = \{R,J\}$ are the absorbing states.  Consider now the generator matrix restricted to the transient states only, which we call $Q = \pp{q_{a,b}}$, where $q_{a,b} = \sigma_{a,b}$ for $a\neq b$ and $q_{a,a} = -\pp{\sigma_a+\lambda_a+\mu_a}$ ($a,b\in\{1,2,\dots,m\}$).Note that the rows of $Q$ do not sum to 0 because of the transitions into absorbing states. The probability of an infective being in infectious state $b$ at time $t$ after their infection if they started from state $a$ can be computed using matrix exponentials as
\begin{equation}
\label{Pstates}
\mathbb{P} \left[ X(t)=I_b \big| X(0) = I_a \right] = \pp{e^{tQ}}_{a,b} \text{ .}
\end{equation}
The rate at which the infective transmits along the edge when is state $I_a$ is $\lambda_a$, so
\begin{equation}
\Laone(t) = \sum_{a=1}^m \lambda_a \ind_{\{X(t)=I_a\}}
\end{equation}
and, from \eqref{BetaDirect},
\begin{eqnarray}
\label{bdef}
\beta(t) & = & \dbmd \E{\Laone(t)} \nonumber \\
& = & \dbmd\sum_{a=1}^{m}\lambda_a \mathbb{P} \Big[ X(t)=I_a \Big| X(0)\sim\nu\Big] \text{ ,}
\end{eqnarray}
or, in vector notation,
\begin{equation}
\label{bdefQ}
\beta(t) = \dbmd\nu \e^{tQ} \lambda^\top \text{ .}
\end{equation}

Note that $\beta(t)$ is described in terms of an exponential of a matrix but that, to calculate $R_0$ and the real-time growth rate $r$, we require only the integral or the Laplace transform of $\beta(t)$. The theory spelt out in \citet{DieHeeBri13}, Section 8.2, considers exactly this situation; therefore, we here only clarify that our definition of $Q$ suitably applies to the assumptions behind such a theory and present the results.

Recall that the spectral bound of a matrix is the largest of the real parts of its eigenvalues. Then, in our case,
\begin{lem}
\label{spectralboundQ}
The spectral bound of $Q$ is strictly negative.
\end{lem}
\begin{proof}
Recall that $Q$ is the generator matrix of the continuous-time Markov chain $\{X(t): t\ge 0\}$ restricted to the transient states $\mathcal{C}$. Note first that $Q$ is weakly diagonally dominant, i.e.~$\abs{q_{a,a}} \ge \sum_{b\ne a}\abs{q_{a,b}}$. Our assumption that individuals recover with probability 1, which is equivalent to all states $I_a$ being transient states, implies that: $Q$ is either irreducible and there is at least one state $a$ such that $\mu_a>0$; or $Q$ is reducible, but any subset of states for which the generator matrix of the Markov chain restricted to that subset is irreducible has at least one state $a$ for which $\mu_a>0$. Consider such a subset of states and denote by $\bar{Q}$ the sub-matrix of $Q$ involving only the states in this subset. Then $\bar{Q}$ is irreducible and is therefore \emph{irreducibly diagonally dominant} and hence invertible (see, e.g.~\citealp{HorJoh85}, Theorem 6.2.27). Further, given $\bar{Q}$ is weakly diagonally dominant and the elements on the diagonal are real and negative, by Gershgorin's circle theorem all its eigenvalues are either null or have strictly negative real part. But invertibility implies 0 cannot be an eigenvalue, so they are all strictly negative and hence the spectral bound of the entire matrix $Q$ is also strictly negative.\qed
\end{proof}

Once we have verified that the spectral bound of $Q$ is negative, Lemma 8.14 of \citet{DieHeeBri13} guarantees that, given the off-diagonal elements of $Q$ (i.e.~the transition rates) are non-negative, the integral $\int_0^t{\e^{u Q}\rd u}$ converges for $t\to \infty$,
\begin{equation}
\label{intconv}
\int_0^\infty{\e^{uQ}\rd u} = -Q^{-1}
\end{equation}
and $-Q^{-1}$ is a positive matrix, whose element in position $(a,b)$ represents the average time spent in state $b$ given we started from state $a$.
Therefore, we have 
\begin{equation}
\label{R0Q}
R_0 = \int_0^\infty{\beta(t)}\:\rd t = \dbmd \nu \pp{-Q^{-1}} \lambda^\top \text{ .}
\end{equation}
To calculate $T_g$ we first use integration by parts to show that
\begin{equation}
\label{intbypart}
\intpos{t\e^{-tQ}\:}{t} = \pp{Q^{-1}}^2\text{ ,}
\end{equation}
from which
\begin{equation}
\label{TgQ}
T_g = \frac{\intpos{t\beta(t)}{t}}{\intpos{\beta(t)}{t}} = \frac{\nu \pp{Q^{-1}}^2 \lambda^\top}{\nu \pp{-Q^{-1}} \lambda^\top} \text{ .}
\end{equation}
Finally, to calculate $r$ from~\eqref{bdefQ} and \eqref{LotkaEuler}, notice that $\e^{tQ}\e^{-rt}$ can be rewritten as $\e^{t(Q-r\mathbf{1}_m)}$, where $\mathbf{1}_m$ is the identity matrix of size $m$. Then, $r$ can be found as the unique solution of
\begin{equation}
\label{rQ}
\dbmd \nu \pp{-\pp{Q-r\mathbf{1}_m}^{-1}}\lambda^\top = 1 \text{ .}
\end{equation}

The theory in \cite{DieHeeBri13}, Section 8.2, also guarantees that \eqref{rQ} has a unique solution and that $R_0 = 1 \iff r = 0$ and $R_0 > 1 \iff r > 0$, as one would expect. 

\subsubsection{Homogeneous mixing}

The theory described above extends to the case of a homogeneously mixing population with with only minor differences. In a large homogeneously mixing population, the probability of making an infectious contact with the same individual twice is negligible, so in this case the Markov chain describing the transmission process requires only the following events and rates:
\begin{equation}\begin{aligned}
I_{a} & \rightarrow I_{b} & \text{ at rate } & \sigma_{a,b} \text{ ,}\\
I_{a} & \rightarrow R & \text{ at rate } & \mu_{a} \text{ .}
\label{redtrans}
\end{aligned}\end{equation}
Again we decompose the state space of the Markov chain in $\mathcal{S}^h =
\mathcal{A}^h\cup\mathcal{C}^h$, where $\mathcal{C}^h=\{I_a\}$ are the transient
states of the system and $\mathcal{A}^h = \{R\}$ is the only absorbing state.
Consider again the generator matrix $Q^h$ restricted to the transient states only. Unlike the network model, where we worked with the rate of making infectious contacts across each single edge, here we specify the rate, $\eta_a$ say, at which an individual in infectious state $a$ makes infectious contacts in the population as a whole and we let $\eta = \pp{\eta_1,\eta_2,\dots,\eta_m}$. The overall rate at which infectious contacts are made in the population is
\begin{equation}
	\beta^h(t) = \sum_{a=1}^m \eta_a \mathbb{P} \left[ X^h(t)=I_a \Big| X^h(0) \sim\nu\right] \text{ ,}
	\label{bhdef}
\end{equation}
or, in vector notation,
\begin{equation}
\label{bhdefQ}
\beta^h(t) = \nu \e^{tQ^h} \eta^\top \text{ .}
\end{equation}

%

Given that the homogeneously mixing and the network cases differ only in the infectivity terms and the definition of the generator of the Markov chain, the same theory as before applies, 
leading to 
\begin{align}
R^h_0 &= \nu \pp{\pp{-Q^h}^{-1}} \eta^\top \text{ ,}\label{R0hQ} \\
T^h_g &= \frac{\nu \pp{\pp{Q^h}^{-1}}^2 \eta^\top}{\nu \pp{\pp{-Q^h}^{-1}} \eta^\top} \text{ ,} \label{TghQ}
\end{align}
and $r^h$ being the unique solution of 
\begin{equation}
\nu \pp{-\pp{Q^h-r\mathbf{1}_m}^{-1}} \eta^\top  = 1 \text{ .}\label{rhQ} \\
\end{equation}

\begin{rem}
One interesting difference between the network and the homogeneously mixing cases is that, on the network, Lemma \ref{spectralboundQ} holds even in the case of an SI model, where individuals never recover. Mathematically, the reason is that, in addition to the recovery rates, the diagonal terms also involve the rates of infecting the neighbour and are therefore non-zero even if the $\mu_a$s are all zero. This implies that $R_0$ is finite even if individuals are infectious forever. The intuition is that, while in the homogeneously mixing model individuals that are infectious forever can lead to $R_0=\infty$, on a network where $\dbmd$ is finite, the number of new infections an infective can make is bounded by the number of neighbours.
\end{rem}

\subsubsection{Relationship with dynamical systems}
The approach we have taken here is slightly unusual, in the sense that many researchers might find it more natural to compute $r$ by linearisation around the disease-free equilibrium of a set of ODEs describing the infection spread. \citet{DieHeeBri13}, Section 7.2, is concerned exactly with discussing the equivalence of these two approaches in the computation of $R^h_0$ in a large and homogeneously mixing population. Adapting the notation of \citet{DieHeeBri13}, we denote by $x^h_a$ the number of individuals in infectious state $I_a$ in the homogeneously mixing population and $x^h=\pp{x^h_1,x^h_2,\dots,x^h_m}$. Then the linear system for the infection dynamics in the early phase of the epidemic can be written as
\begin{equation}
\label{LinSysh}
\frac{\rd}{\rd t} x^h = x^h (T^h+Q^h)\text{ ,}
\end{equation}
where $T^h$ is a matrix containing only those terms related to transmission and $Q^h$, as before, describes the transitions between states (including recovery rates, which cause some of the rows not to sum to 0). The $m\times m$ transmission matrix has the special form $T^h=\nu^\top \eta$ because an infective in state $a$ generates new cases in state $b$ at rate $\eta_a\nu_b$, and has therefore rank 1, with the only non-zero eigenvalue given by $\nu \eta^\top$. \citet{DieHeeBri13} show how $R_0$ can then be calculated as the dominant eigenvalue of $ -T^h \pp{Q^h}^{-1}$, which in our case leads to \eqref{R0hQ}. The computation of $r$ follows the same argument once $Q^h$ is replaced by $(Q^h-r\mathbf{1})$, as can be noted by comparing Equations \eqref{R0hQ} and \eqref{rhQ} above.

The case of the network model is slightly more complicated, because in addition to the infectious state of individuals, their infectivity also depends on their degree and the number of edges still available for transmission (the used ones are `burnt out' and do not allow further transmission). Therefore, we need to write a linear system for the vector $x = \pp{x_a^{(d,s)}}$ of the numbers of individuals in state $a$ with degree $d$ and $s$ susceptible neighbours ($a\in\{1,2,\dots,m\}$, $d=1,2,\dots$ and $s\in\{0,1,\dots,d-1\}$). Again, the system can be put in the form (assuming for simplicity a finite maximum degree)
\begin{equation}
\label{LinSys}
\frac{\rd}{\rd t} x = x(T+\Sigma) \text{ ,}
\end{equation}
for a suitable transmission matrix $T$ and transition matrix $\Sigma$, and $R_0$ can be calculated as the dominant eigenvalue of $K_L = -T\Sigma^{-1}$. The matrix $K_L$, which can in general have high dimensionality, is referred to in \citet{DieHeeBri13} as the next-generation matrix \emph{with large domain}. The authors, however, explain clearly that one can get a first dimensionality reduction by focusing only on those states a newly infected individual can start from (here, only those states where $s=d-1$, the decrease in $s$ being only a transition due to the infective making infectious contacts with neighbours). Furthermore, a second dimensionality reduction occurs when new cases are generated in fixed proportions, so that, stochastically speaking, there is only one type of new infective, which starts in each state according to a fixed distribution. In our case, any new case has degree $d$ and starts in  infectious state $a$ with probability $\nu_a d \, \prob{D=d} / \E{D}$, so that, in fact, the calculation of $R_0$ (and analogously of $r$) can be reduced to a 1-dimensional equation; see \eqref{R0Q} and \eqref{rQ}. This last simplification is the same one we applied in the homogeneously mixing case when we noticed that the matrix $T^h$ has rank 1, which is the case exactly because new infectives start in each state in fixed proportions given by $\nu$.

\subsubsection{Efficient numerical methods}

We now present a numerically efficient algorithm for the computation of $r$ that avoids the costly and potentially numerically unstable matrix inversion. To facilitate application, we express it directly in terms of the transition rates of the original compartmental model.
\begin{prop}
	For a compartmental model as defined above, $r$ is given by the largest real solution to
	the equations
\begin{align}
\lambda_a - (r+\sigma_a +\lambda_a + \mu_a)z_a(r) + 
\sum_{b \neq a} \sigma_{a,b} z_b(r)  & = 0 \text{ ,} 
\label{pi}
\\
\dbmd\sum_{a=1}^m \nu_a z_a(r) & = 1\text{ .}
\label{piE}
\end{align}
\end{prop}

\begin{proof}
	First, define 
	\begin{equation}
z_a(r) := \mathbb{E}\bigg[
\int_0^{\infty} \Laone(t) \e^{-rt} \mathrm{d} t
\ \bigg| \ Y(0) = I_a \bigg] \text{ .}
\end{equation}
Then~\eqref{pi} follows (after a little algebra) from Proposition 2
of~\citet{Pollett:2002}.  Combining~\eqref{BetaDirect} and~\eqref{LotkaEuler},
we find that $r$ satisfies
\begin{equation}
	\begin{aligned}
		1 & = \dbmd \int_0^\infty\E{\Lambda_1(t)}\e^{-rt}\:\mathrm{d}t \\
			& = \dbmd \E{\int_0^\infty\Lambda_1(t)\e^{-rt}\:\mathrm{d}t} \\
		 & = \dbmd\sum_{a=1}^{m} \mathbb{E}\bigg[
\int_0^{\infty} \Laone(t) \e^{-rt} \mathrm{d} t
\ \bigg| \ Y(0) = I_a \bigg] \mathbb{P} \left[ Y(0) = I_a \right]\\
& = \dbmd\sum_{a=1}^m \nu_a z_a(r) \text{ ,} \label{Eunif}
	\end{aligned}
\end{equation}
which establishes~\eqref{piE}. The facts that there is a unique such value of $r$ with largest real part and that it is real lie in the theory developed above.
\end{proof}

\begin{prop}
	For homogeneous mixing, $r^h$ is given by the largest real solution to
	the equations
\begin{align}
\lambda_a - (r^h +\sigma_a + \mu_a)z^h_a(r^h) + 
\sum_{b \neq a} \sigma_{a,b} z^h_b(r^h)  & = 0 \text{ ,} 
\label{pihom1}
\\
\sum_{a=1}^m \nu_a z^h_a(r^h) & = 1\text{ .}
\label{pihom2}
\end{align}
\end{prop}

\begin{proof}
	Define
	\begin{equation}
z^h_a(r^h) := \mathbb{E}\bigg[
\int_0^{\infty} \beta^h(t) \e^{-r^h t} \mathrm{d} t
\ \bigg| \ Y^h(0) = I_a \bigg] \text{ ,}
\end{equation}
use~\eqref{bhdef} and follow the analogous reasoning to~\eqref{Eunif} above.
\end{proof}

\subsubsection{Examples}
\label{sec:ExamplesComp}

We consider four specific compartmental structures, for which we explicitly
state only the non-zero rates, and where the $\lambda$ rates apply to the network model and the $\eta$ rates to the homogeneously mixing model. For the first three of these structures, we do not have
different infectivities for infectious states, but instead use compartments to
model different recovery time distributions. There is quite an extensive theory
of the `phase-type' distributions that arise from use of compartments in this
way. In particular, it is known that they are dense in the space of
positive-valued probability distributions \citep{Neuts:1975} and so the present
approach can be used, at least in theory, to approximate any desired
distribution within any required level of accuracy. The fourth structure is the
standard SEIR model.  In general, models with $m$ compartments will lead to
$m$-th order polynomials, with the attendant restrictions on what can be said
analytically.  Nevertheless, as indicated above the values of $r$ can be
calculated numerically using an efficient algorithm.  

\paragraph{Markovian dynamics}

These have already been considered above; however, for completeness we note
that these would be defined through
\begin{equation}
\nu_1 = 1 \text{ ,}\quad \mu_{1} = \gamma \text{ ,} \quad \lambda_{1} =\lambda \quad \text{ and }\quad \eta_1 = \eta \text{ ,} 
\end{equation}
and we obtain
\begin{align*}
	R_0^h & = \frac{\eta}{\gamma} \text{ ,} &
	T_g^h & = \frac{1}{\gamma} &
	r^h & = \eta - \gamma \text{ ,} \\
	R_0 & = \dbmd \frac{\lambda}{\gamma + \lambda} \text{ ,} &
	T_g & = \frac{1}{\gamma+\lambda} &
	r & = (\dbmd-1)\lambda - \gamma \text{ ,}
\end{align*}
in agreement with what shown in Section \protect\ref{sec:MarkovianSIRModel}.

\paragraph{Hypo-exponential}

The hypo-exponential distribution has less variability than the
exponential and consists of a chain of phases. For example, consider the special case of two equally infective compartments in series
\begin{equation}
\nu_1 = 1 \text{ ,} \quad \sigma_{1,2} = \gamma_1
\text{ ,} \quad\mu_{2} = \gamma_2 \text{ ,}\quad \lambda_{1} = \lambda_{2} = \lambda \quad \text{ and }\quad \eta_{1} = \eta_{2} = \eta \text{ .}
\end{equation}
In this case, we obtain
\begin{align*}
	R_0^h & = \eta \pp{\frac{1}{\gamma_1}+\frac{1}{\gamma_2}} \text{ ,} \qquad \qquad \qquad \qquad \qquad T_g^h = \frac{1}{\gamma_1}+\frac{1}{\gamma_2} - \frac{1}{\gamma_1 + \gamma_2}  \text{ ,} \\
r^h & = \frac12 \left(
\eta - (\gamma_1 + \gamma_2) + \sqrt{%
	(\gamma_1 - \gamma_2)^2 + \eta^2
	+2 (\gamma_1 + \gamma_2) \eta
}\right)\text{ ,}\\
	R_0 & = \dbmd \left(1- \frac{\gamma_1 \gamma_2}%
{(\lambda +\gamma_1)(\lambda +\gamma_2)} \right)\text{ ,} \qquad \quad \quad T_g = \frac{1}{\lambda + \gamma_1}+\frac{1}{\lambda + \gamma_2} - \frac{1}{\lambda + \gamma_1 + \gamma_2}  \text{ ,}  \\
r & = \frac12 \left(
(\dbmd-2)\lambda - (\gamma_1 + \gamma_2) + \sqrt{%
(\dbmd \lambda + \gamma_2 - \gamma_1)^2 + 4 \dbmd \lambda \gamma_1
}\right)\text{ ,}
\end{align*}
which, in the special case of $\gamma_1 = \gamma_2 = \gamma$ leads to the same result that is derived from Section \ref{sec:GammaDistributedInfectiousPeriod} when $I$ is Erlang-distributed with $\alpha=2$ and scale parameter $\gamma$.

\paragraph{Hyper-exponential}

The hyper-exponential distribution has more variability than the
exponential and consists of a parallel set of phases. For example, consider the special case of two equally infective compartments in parallel, with different sojourn times
\begin{equation}
\nu_1 = \nu_2 = 1/2 \text{ ,}\quad \mu_{1} = \gamma_1
\text{ ,}\quad \mu_{2} = \gamma_2  \text{ ,}\quad \lambda_{1} = \lambda_{2} = \lambda \quad \text{ and }\quad \eta_1 = \eta_2 = \eta \text{ ,}
\end{equation}
we obtain
\begin{align*}
		R_0^h & = \frac\eta{2} \pp{\frac1{\gamma_1} + \frac1{\gamma_2}} \text{ ,} \qquad \qquad \; \quad \quad T_g^h = \frac{1}{\gamma_1}+\frac{1}{\gamma_2} - \frac{2}{\gamma_1 + \gamma_2}  \text{ ,} \\
r^h & = \frac12 \left(
\eta - (\gamma_1 + \gamma_2) + \sqrt{%
	(\gamma_1 - \gamma_2)^2 + \eta^2
}\right) \text{ ,}\\
	R_0 & = \dbmd\frac{\lambda(\lambda + \frac12 (\gamma_1 + \gamma_2))}%
{(\lambda +\gamma_1)(\lambda +\gamma_2)} \text{ ,} \qquad \quad \quad T_g = \frac{1}{\lambda + \gamma_1}+\frac{1}{\lambda + \gamma_2} - \frac{2}{2\lambda + \gamma_1 + \gamma_2}  \text{ ,} \\
r & = \frac12 \left(
(\dbmd-2)\lambda - (\gamma_1 + \gamma_2) + \sqrt{%
(\dbmd\lambda )^2 + (\gamma_1 - \gamma_2)^2
}\right) \text{ .}
\end{align*}
Note that in this case the transition matrix is reducible (see Lemma \ref{spectralboundQ}).
\paragraph{SEIR}

The standard SEIR model obeys
\begin{equation}
\nu_1 = 1 \text{ ,}\quad \sigma_{1,2} = \delta
\text{ ,} \quad\mu_{2} = \gamma \text{ ,} \quad \lambda_{2} = \lambda \quad \text{ and }\quad \eta_2 = \eta \text{ ,}
\end{equation}
where $\delta$ represents the rate of transition from the latent to the infectious state and $\gamma$ the recovery rate. We obtain
\begin{align*}
		R_0^h & = \frac{\eta}{\gamma}  \text{ ,} \qquad \qquad \qquad \quad T_g^h = \frac{1}{\delta} + \frac{1}{\gamma}  \text{ ,}\\
r^h & = \frac12 \left(
\sqrt{%
	(\gamma - \delta)^2 + 4 \eta \delta
} - (\gamma + \delta)\right) \text{ ,}\\
	R_0 & = \dbmd \frac{\lambda}{\gamma + \lambda} \text{ ,} \qquad \qquad T_g = \frac{1}{\delta} + \frac{1}{\lambda + \gamma}  \text{ ,}\\
r & = \frac12 \left(
\sqrt{%
	(\lambda + \gamma + \delta)^2 + 4 \delta((\dbmd-1)\lambda - \gamma) 
} - (\lambda + \gamma + \delta)\right) \text{ .}
\end{align*}
Note that combining the first two equations for the homogeneously, one obtains the known expression relating $R_0^h$ and $r^h$ in the SEIR model,
\begin{equation}
 \label{R0hrhSEIR}
 R_0^h = \pp{1+\frac{r^h}{\delta}}\pp{1+\frac{r^h}{\gamma}}\text{ .}
 \end{equation}

\section{A general approach via direct calculation}
\label{sec:OtherRandomTVIModels}

In this section we examine a general approach to the calculation of the real-time growth rate $r$, by direct calculation of the distribution of the time taken for infection to pass across an edge, $W_1$. Once the distribution of $W_1$ and its pdf $w_1(t)$ have been determined we can easily calculate $T_g=\E{W_1}$ and $\beta(t)=\dbmd\E{1-\e^{-A}}w_1(t)$, leading to the real time growth rate $r$ via the Lotka-Euler equation \eqref{LotkaEuler}.

Let $\Omega$ be the sample space of the infectivity profile distribution and let $\Lambda=\{\Lambda(t):t\geq0\}\subseteq\Omega$ be the family of random variables representing an infectivity profile drawn from this distribution. Recall that infectious contacts are made between an infected individual and each of their neighbours at the points of an inhomogeneous Poisson process with rate $\Lambda(t)$. Let $M$ represent the total number of contacts made across an edge, and hence $M\sim\mathrm{Pois}(A)$, where $A=\int_0^\infty\Lambda(t)\:\mathrm{d}t$. Recall that $W$ is a random variable representing the times at which infectious contacts are made and $W_1$ represents the time of the first of these infectious contacts. Clearly these two random variables are conditional on at least one infectious contact occurring at some point, i.e.~$M>0$.

\begin{prop}\label{lemma1}
The equation below links the cdf of $W_1$ to the cdf of $W\!$ given $\Lambda$.
\begin{equation}\label{lemma1eq}
\mathbb{P}(W_1\leq x|M>0)=1-\frac{\mathbb{E}_\Lambda\!\left[\e^{-A\mathbb{P}(W\leq x|M>0,\Lambda)}-\e^{-A}\right]}{\mathbb{E}_\Lambda\!\left[1-\e^{-A}\right]},\nonumber
\end{equation}
where $A=\int_0^\infty\Lambda(t)\:\mathrm{d}t$.
\end{prop}

\begin{proof}
\begin{align*}
\mathbb{P}(W_1\leq x|M>0)&=
\frac{\mathbb{P}(W_1\leq x,M>0)}{\mathbb{P}(M>0)}
\end{align*}
Conditioning on the joint distribution of $M$ and $\Lambda$ yields
\begin{align*}
\mathbb{P}(W_1\leq x|M>0)&=
\frac{\int_\Omega\sum_{m=1}^\infty\mathbb{P}(W_1\leq x|M=m,\Lambda)\:\mathrm{d}\mathbb{P}(M=m,\Lambda)}
{\mathbb{P}(M>0)}\\
&=\frac{\mathbb{E}_\Lambda\!\left[\sum_{m=1}^\infty\mathbb{P}(W_1\leq x|M=m,\Lambda)\mathbb{P}(M=m|\Lambda)\right]}
{\mathbb{P}(M>0)}
\end{align*}
Given $M=m>0$, the first contact time $W_1$ is the minimum of $m$ iid realisations of $W$, and so $\mathbb{P}(W_1\leq x|M=m,\Lambda)=1-\mathbb{P}(W>x|M>0,\Lambda)^m$. 
\begin{align*}
\mathbb{P}(W_1\leq x|M>0)&=
\frac{\mathbb{E}_\Lambda\!\left[\sum_{m=1}^\infty(1-\mathbb{P}(W>x|M>0,\Lambda)^m)\mathbb{P}(M=m|\Lambda)\right]}
{\mathbb{P}(M>0)}\\
&=1-\frac{\mathbb{E}_\Lambda\!\left[\sum_{m=1}^\infty\mathbb{P}(W>x|M>0,\Lambda)^m\mathbb{P}(M=m|\Lambda)\right]}
{\mathbb{E}_\Lambda\!\left[\mathbb{P}(M>0|\Lambda)\right]}
\end{align*}
Finally recall that $M|\Lambda\sim\mathrm{Pois}(A)$ and hence $\E{z^M|\Lambda}=\e^{-A(1-z)}$ and $\mathbb{P}(M=0|\Lambda)=\e^{-A}$.
\end{proof}

If the cdf of $W$ is known then Proposition \ref{lemma1} can be used to determine the cdf of $W_1$. 

\subsection{Exponential total infectivity}


In this section we consider models in which the total infectivity $A=\int_0^\infty\Lambda(t)\:\mathrm{d}t$ is exponentially distributed. Important examples in this class of models includes the sSIR model with exponential infectious periods, the TVI model where $\Lambda(t)=Af(t)$, where $A$ is exponentially distributed and $f$ is any proper pdf.

\begin{prop}\label{lemma2}
If $A$ is exponentially distributed then:
\begin{enumerate}
\item $M\sim\mathrm{Geom}(p)$, where $p=1/(\E{A}+1)$ and $\mathbb{P}(M=m)=p(1-p)^m$ for $m=0,1,\dots$.
\item $M|(M>0)\sim 1+\mathrm{Geom}(p)$.
\item The cdf of $W_1$ is given by
\begin{equation}
\mathbb{P}(W_1\leq x)=1-\frac{p\mathbb{P}(W>x)}{1-(1-p)\mathbb{P}(W>x)}\label{family}
\end{equation}
\end{enumerate}
\end{prop}

\begin{proof}
\begin{enumerate}
\item $M|A\sim\mathrm{Pois}(A)$ and removing the conditioning yields $\mathbb{P}(M=m)=p(1-p)^m$.
\item This follows directly from the lack-of-memory property of the geometric distribution. 
\item 
\citet{Mar97} 
show that, given any iid sequence of random variables $X_1,X_2,\dots$ with cdf $F_X(x)$ and $N\sim1+\mathrm{Geom}(p)$, the cdf of $U=\min\{X_1,\dots,X_N\}$ satisfies $1-F_U(x)=\frac{p(1-F_X(x))}{1-(1-p)(1-F_X(x))}$. By 2, $W_1$ is the minimum of a geometrically distributed number of random variables, distributed according to $W$.
\end{enumerate}
\end{proof}

\citet{Mar97} considers families of random variables $\mathcal{F}$ in which the minimum of a random number $N$ of iid variables is also within the family $\mathcal{F}$. The authors show that in such cases $N$ \emph{must} be geometrically distributed and the family is characterised by \eqref{family} with $0<p<1$.

\section{Numerical method}
\label{sec:NumericalMethod}

It is possible to calculate the real-time growth rate $r$ numerically for the most general model described in this paper by following the procedure detailed below. First we use Monte Carlo sampling with Poisson thinning to obtain an estimate for $\bbE[\Laone(t)]$ and hence $\beta(t) = \dbmd \E{\Laone(t)}$; then we solve Equation \eqref{LotkaEuler} iteratively for $r$.

We begin by choosing a grid of $G$ values, $t_1 < t_2 < \dots < t_G$, over which we will apply a quadrature method to compute the integral in \eqref{LotkaEuler} (e.g.~the trapezium rule). It is important that the grid covers the range of $t$ for which the integrand has a significant contribution to the integral. Finding a suitable value for $t_G$ is usually not a problem because, since $\int_0^\infty\bbE[\Laone(t)]\:dt$ is finite, then $\bbE[\Laone(t)]$ tends to zero as $t\rightarrow\infty$ and, in the most interesting case of $r>0$, so does $\e^{-rt}$.

Next we simulate $M$ samples from the infectious contact distribution. For the $i$-th sample $\{\Lambda^{(i)}(t):t\geq 0\}$, first calculate $m_i=\sup\{\Lambda^{(i)}(t):t\geq 0\}$ and then simulate events from a homogeneous Poisson process with rate $m_i$, for example using the Gillespie algorithm. This algorithm begins at time 0, and the time until the next event is drawn sequentially from an exponential distribution with rate $m_i$. To obtain a sample of events from the inhomogeneous Poisson process with rate $\Lambda^{(i)}(t)$, we accept event $j$ at time $\tau_j$ in the homogeneous Poisson process with probability $\Lambda^{(i)}(\tau_j)/m_i$, and reject the event otherwise. Let $\tau^{(i)}_1$ denote the time of the first accepted event, that is the time of the first event in the inhomogeneous Poisson process with rate $\Lambda^{(i)}(t)$. For more details regarding this Poisson thinning procedure see \citet{Kin92}. We need only to continue simulating events in the homogeneous Poisson process until either the first event is accepted or the time of the first accepted event $\tau^{(i)}_1$ is known to be after the last of our grid points, $t_G$. Finally to complete the calculation of $\bbE[\Laone(t)]$, we use the fact that
\begin{align*}
\bbE[\Laone(t)]&=\bbE[\Lambda(t)\bbP(\text{no event in }(0,t)|\Lambda)] \\
&=\bbE[\Lambda(t)\bbE[\ind_{\{\text{no event in }(0,t)\}}|\Lambda]] \\
&\approx\sum\limits_{i=1}^M\Lambda^{(i)}(t)\ind_{\{\tau^{(i)}_1<t\}}
\end{align*}
and hence that
\[ \beta(t)\approx \dbmd 
\sum\limits_{i=1}^M\Lambda^{(i)}(t)\ind_{\{\tau^{(i)}_1<t\}}. \]

Note that we can use the same set of simulated $\Lambda^{(i)}(t)$ and $\tau^{(i)}_1$ to calculate the integrand in \eqref{LotkaEuler} for all of our grid points $t_1,\dots,t_G$ and therefore the computation time needed for the first part of the procedure does not depend on the number of grid points $G$.

The second part of the algorithm involves the use any iterative method (e.g.~the secant method; we used MATLAB\textsuperscript{\textregistered} built-in function \texttt{fzero}) coupled with any quadrature method (e.g.~the trapezium rule) to find $r$ from Equation \eqref{LotkaEuler}.


\section{Numerical results}
\label{sec:NumericalResults}

In this section we explore the key epidemiological quantities in both the homogeneously mixing and the network models, for various choices of the network degree distribution and the infectivity profile. When a TVI model is used we assume a gamma-shaped infectivity profile, i.e.~proportional to the pdf of a gamma distribution, where the shape parameter $\alpha$ takes values of 0.5, 1, 2, 10 and the limiting case of $\alpha\to\infty$, which corresponds to the Reed-Frost model where all the infectivity is spread at time $T_g^h=T_g$. When a sSIR model is assumed, the duration $I$ of the infectious period is assumed to follow a gamma distribution, again with shape parameter $\alpha = 0.5, 1, 2, 10$ and the limiting case of $\alpha\to\infty$, which corresponds to a fixed duration infectious period. The choices of $\alpha$ in both cases, correspond to the gamma distributions where the ratio between the variance and the mean is 2, 1, 1/2, 1/10 and 0, respectively. 

\subsection{Network model outputs}
\label{sec:NetworkModelOutputs}
In Figures \ref{fig:TVIoutput} and \ref{fig:sSIRoutput} we plot three network-related outputs -- $T_g$ (left column), $R_0$ (middle column) and $r$ (right column) -- as a function of the basic reproduction number for the corresponding homogeneously mixing model (see section \ref{sec:Comparisons}). Figure \ref{fig:TVIoutput} examines the TVI model (with fixed $R_0^h$ and $T_g^h$) and Figure \ref{fig:sSIRoutput} examines the sSIR model (at fixed $R_0^h$ and $\E{I}$) for different network degree distributions: regular (top row), Poisson (middle row) and negative binomial with variance 5 times larger than the mean (bottom row). Together with $r$, we also show the approximation described in Section \ref{sec:Approximater}, for each infectivity profile.

Note how the value of $R_0^h$ unequivocally determines the value of $R_0$ when the infectivity profile is non-random (Figure \ref{fig:TVIoutput}, central column), but not when it is random (Figure \ref{fig:sSIRoutput}, central column). Other observations are: $T_g$ and $T_g^h$ are always identical in the Reed-Frost model; increasing the infectivity leads to shorter generation times on the network and larger values of the real-time growth rate (with the exception of small infectivities for very overdispersed durations of infection in the sSIR model, e.g.~when $\alpha<1$). Finally, the strongest effects of infection interval contraction and repeated contacts, i.e.~shorter $T_g$, lower $R_0$ and $r$, and therefore also a less accurate approximation, are experienced when $\dbmd$ is smallest: at fixed mean degree, this is the case of the regular network (zero variance); at fixed distribution, for smallest mean degree (Figure \ref{fig:Impactofnetworks}). With the latter in mind, we chose here a mean degree of $\E{D} = 3$, which is large enough to allow the existence of a giant connected component, but small enough to accentuate saturation effects.

\begin{figure}
\centering
\includegraphics[width = \textwidth]{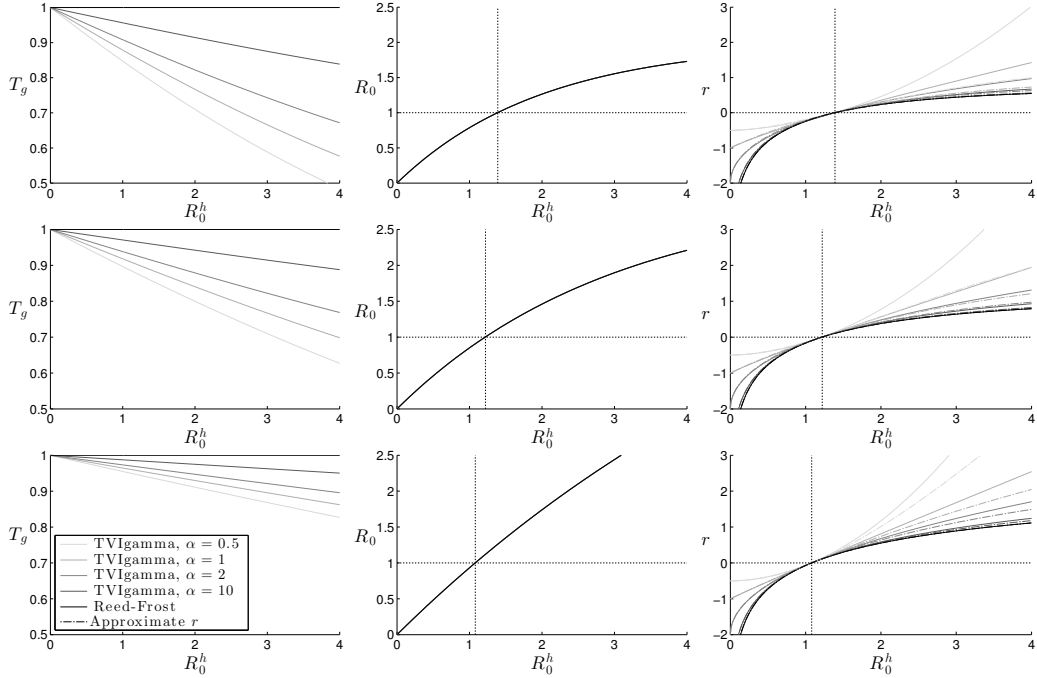}
\caption{The network-related outputs (the generation time $T_g$, first column; the basic reproduction number $R_0$, second column; and the real-time growth rate $r$, together with its approximation $r_{\text{app}}$, third column) as a function of the basic reproduction number of the corresponding homogeneously mixing model, $R_0^h$, for the TVI model with gamma-shaped infectivity profiles having shape parameter $\alpha=0.5,1,2,10$ and $\infty$ (the Reed-Frost model). Each row considers a different degree distribution: regular ($\Var{D}=0$), first row; Poisson ($\Var{D}=\E{D}$), second row; and negative binomial (for which we imposed $\Var{D}=5\E{D}$), third row. All networks have mean degree $\E{D}=3$, a reasonably small value to accentuate the network saturation effects on disease spread. All infectivity profiles match the same value of $T_g^h=1$. The black vertical dotted line indicates the value of $R_0^h$ for which $R_0=1$, which depends on the degree distribution. The horizontal black dotted line shows where $R_0=1$ (middle column) or $r=0$ (right column).}
\label{fig:TVIoutput}
\end{figure}

\begin{figure}
\centering
\includegraphics[width = \textwidth]{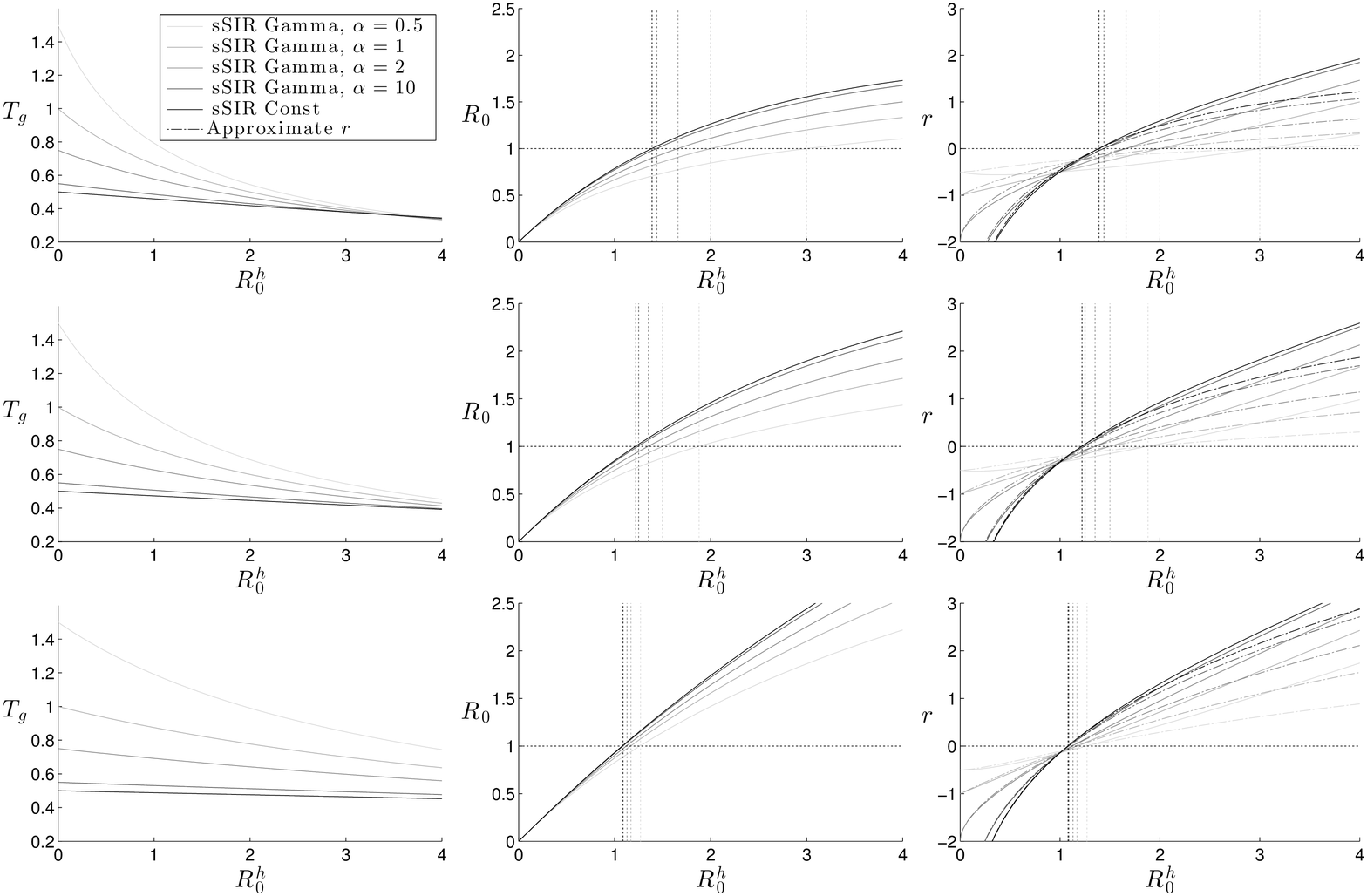}
\caption{Network-related quantities (the generation time $T_g$, first column; the basic reproduction number $R_0$, second column; and the real-time growth rate $r$, together with its approximation $r_{\text{app}}$, third column) as a function of the basic reproduction number of the corresponding homogeneously mixing model, $R_0^h$, for the sSIR model with gamma-shaped infectious period distributions having shape parameter $\alpha=0.5,1,2,10$ and $\infty$ (i.e.~constant infectious period). Each row considers a different degree distribution: regular ($\Var{D}=0$), first row; Poisson ($\Var{D}=\E{D}$), second row; and negative binomial (in which $\Var{D}=5\E{D}$), third row. All infectivity profiles have infectious periods with $\E{I}=1$ and no latent period. All networks have mean degree $\E{D}=3$, a reasonably small value to accentuate the network saturation effects on disease spread. The vertical dotted lines indicate the value of $R_0^h$ for which $R_0=1$ for each infectivity profile (corresponding shades of grey). The horizontal black dotted line shows where $R_0=1$ (middle column) or $r=0$ (right column).}
\label{fig:sSIRoutput}
\end{figure}

The impact of the network structure is highlighted in Figure \ref{fig:Impactofnetworks}, where four networks with different degree distributions (regular, Poisson, geometric and negative binomial with variance 5 times larger than the mean) are investigated. The real-time growth rate was plotted as a function of the mean degree $\E{D}$. Except for the geometric degree distribution, for all other networks the ratio $\Var{D}/\E{D}$ is constant and the magnitude of this constant is reflected, for any chosen infectivity profile, in the clear ordering of $r$ from regular to Poisson to negative binomial. The comparison between different networks and infectivity profiles is performed as follows. All infectivity profiles are from the sSIR model and have the same value of $\E{I}$. We fixed $r^h=0.5$, and from this the corresponding values of $R_0^h$ and $\E{A}$ were obtained, which were then used to compute $r$. Note how all network models converge to the homogeneously mixing model in the limit of large mean degree. Finally, if we were to plot $r$ as a function of $R_0$ (not shown) then the choice of the degree distribution would be irrelevant, because the only impact of $D$ is in the factor $\dbmd$, which appears identically in the computation of both $R_0$ and $r$. In other words, the relationship between $R_0$ and $r$ given by Equation \protect\eqref{rR0rel} only depends on $W_1$ (i.e.~on $w_1(t)$), which is affected by the choice of the infectivity profile, but not by the overall mean infectivity or network structure.

\begin{figure}
\centering
\includegraphics[width = 0.9\textwidth]{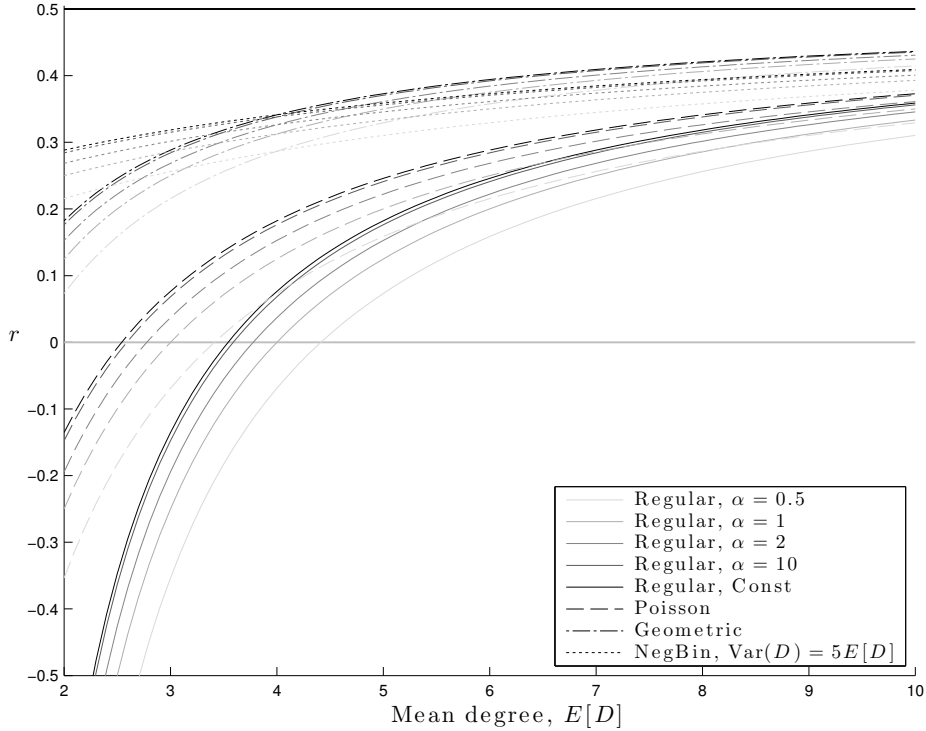}
\caption{Impact of the network degree distribution: $r$ as a function of the mean degree for fixed $r^h=0.5$, using the sSIR model with gamma distributed infectious periods having shape parameter $\alpha = 0.5,1,2,10$ and $\infty$ (i.e.~constant duration) and four different degree distributions: regular ($\Var{D} = 0$, continuous line), Poisson ($\Var{D} = \E{D}$, dashed), geometric ($\Var{D} = \E{D} ( 1+\E{D})$, dash-dotted) and negative binomial with $\Var{D} = 5 \E{D}$ (dotted). All infectivity profiles share $\E{I} = 1$ and $r^h = 0.5$. The horizontal thick black and grey lines show $r=0.5$ and $r=0$. For increasing mean degree all network models converge to the homogeneously mixing model and the smaller mean degrees are associated with stronger network saturation effects and hence slower epidemic spread.}
\label{fig:Impactofnetworks}
\end{figure}

\subsection{Model comparisons}
\label{sec:ModelComparisons}
Figure \ref{fig:ModelComparison} highlights the importance of being unambiguous about what is kept fixed when a  model comparison is performed. Even in the absence of network structure, Figures \ref{fig:ModelComparison}a and \ref{fig:ModelComparison}b show how, for the sSIR model with $R_0^h>1$, $r^h$ can increase or decrease as the variance of the infectious period duration $I$ decreases (the shape parameter $\alpha$ increases) when either $\E{I}$ or $T_g^h$ are kept fixed, respectively. On the other hand, Figures \ref{fig:ModelComparison}c and \ref{fig:ModelComparison}d show how, for the sSIR model on a network (regular with degree 3 in this example), $r$ can diverge as $R_0$ approaches $\dbmd$ or not, depending of whether $\E{I}$ or $T_g$ is kept fixed, respectively. The approximation described in Section \ref{sec:Approximater} is also shown, and it is found to be particularly inaccurate because this example tests it at the limit when effects of repeated contacts and infection interval contraction are strongest. 

The intuitive reason why $r$ diverges when $\E{I}$ is kept fixed (Figure \ref{fig:ModelComparison}c) is that, as $R_0\to \dbmd$, the infection rate $\lambda\to\infty$, which means that the time at which the first infectious contact occurs shrinks to 0 and thus $r\to\infty$. Because the approximation ignores precisely this aspect (the time of the first infectious contact is approximate with the time of a randomly selected contact), unlike $r$, the approximation $r_{\text{app}}$ does not diverge. Furthermore, as $\lambda\to\infty$ the generation time $T_g$ shrinks to 0. Therefore if $T_g$ is kept fixed instead of $T_g^h$ (as in Figure \ref{fig:ModelComparison}d), $\E{I}$ and $T_g^h$ diverge, but the resulting $r$ does not (because the time of the first infectious contact does not shrink to 0). The approximation $r_{\text{app}}$ even decreases because the average time of an infectious contact diverges. The same phenomenon of $r$ exploding as in Figure \ref{fig:ModelComparison}c tends to occur also for the TVI model (not shown) although as the shape of the infectivity profile becomes more peaked around its mode, the infection rate just after infection decreases. This prevents $T_g$ from getting too close to 0. In the limiting case of the Reed-Frost model $T_g=T_g^h$, and therefore $r=r_{\text{app}}$ and neither of them diverge (nor decrease) when $R_0\to\dbmd$.

\begin{figure}
\centering
\includegraphics[width = 0.5\textwidth]{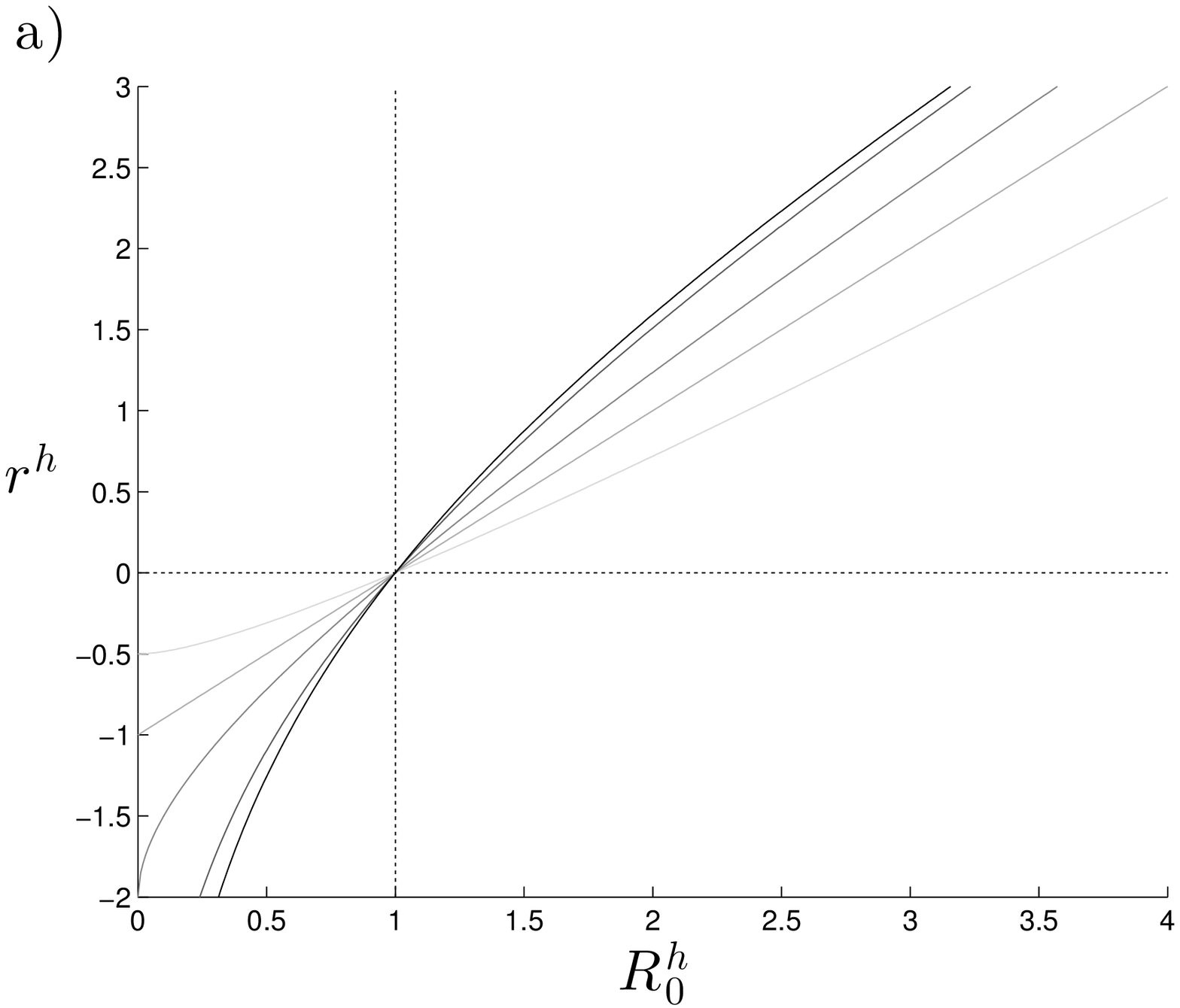}\includegraphics[width = 0.5\textwidth]{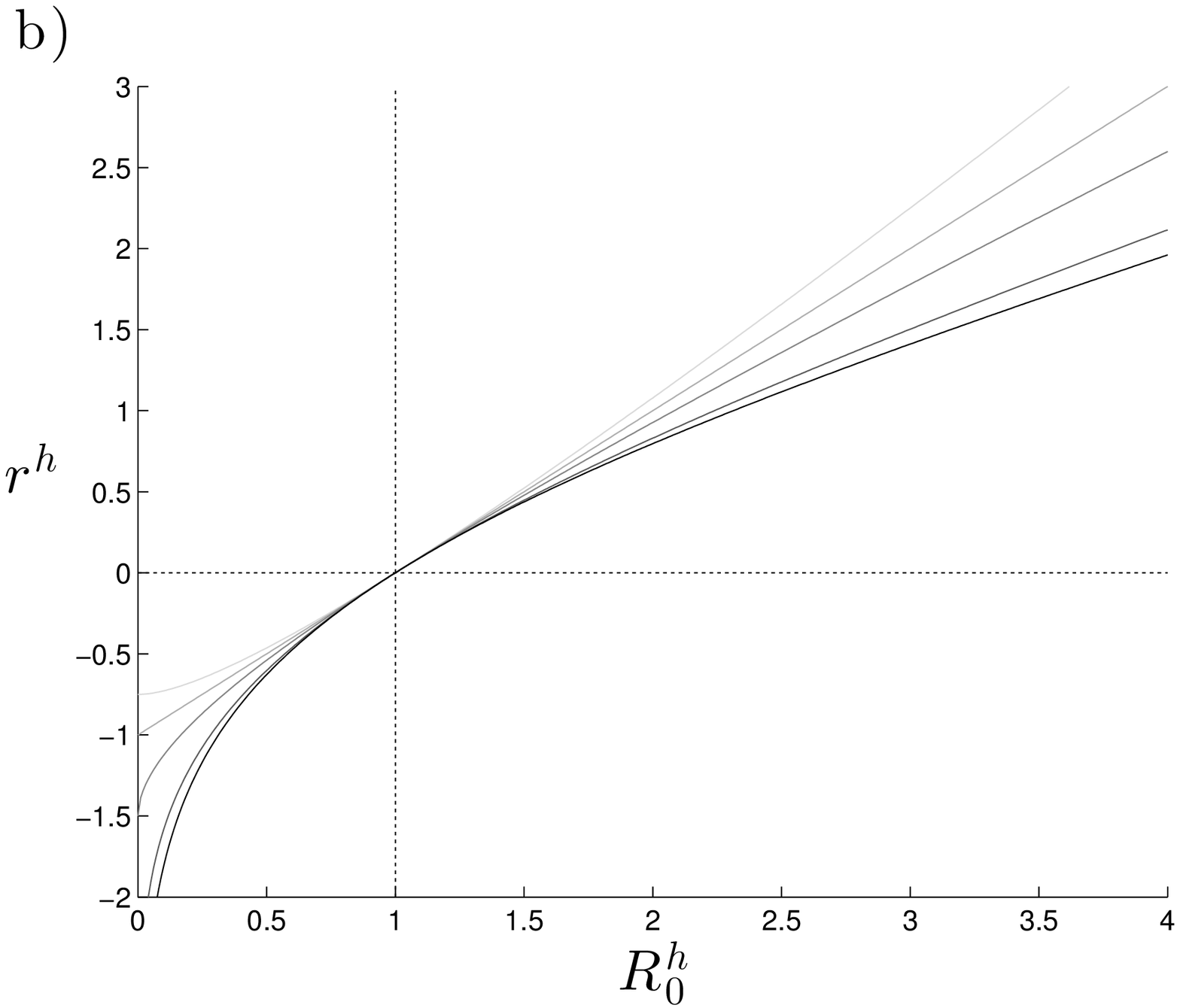}
\includegraphics[width = 0.5\textwidth]{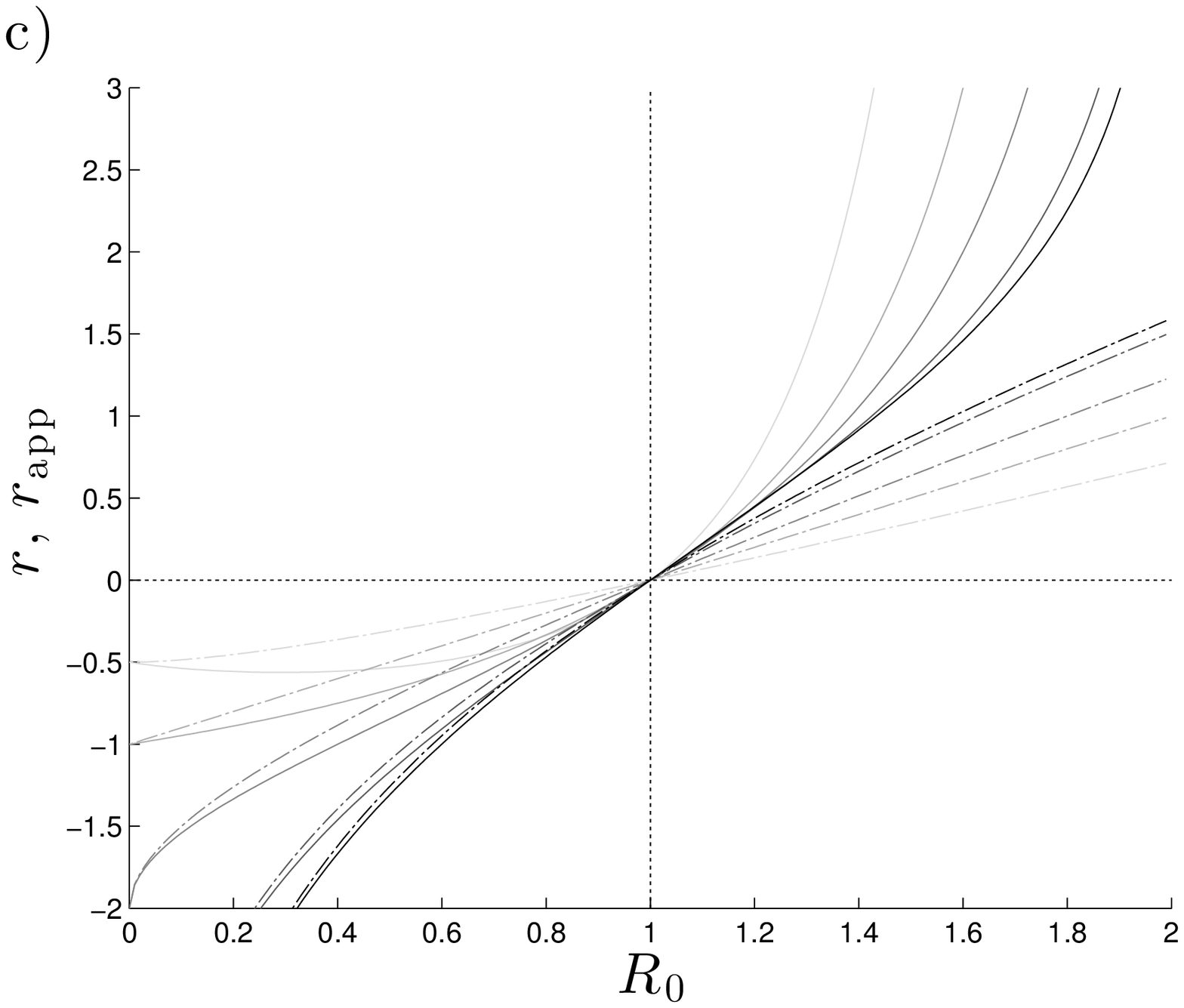}\includegraphics[width = 0.5\textwidth]{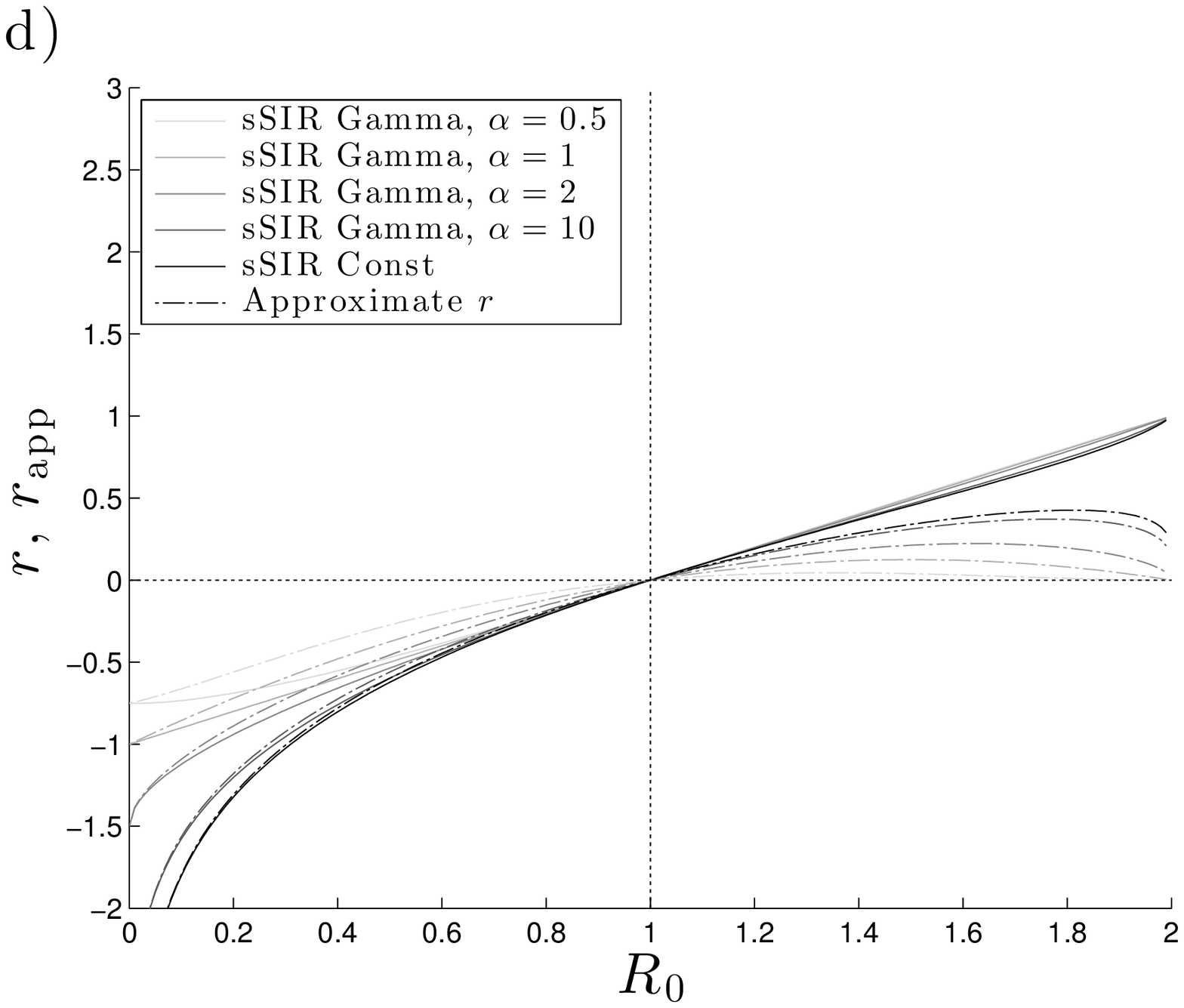}
\caption{Illustration of which quantities to keep fixed when the impact of different infectivity profiles is compared. First row: relationship between $r^h$ and $R_0^h$ in a purely homogeneously mixing population and sSIR infectivity profile with gamma distributed duration of the infectious period $I$ having shape parameter $\alpha = 0.5,1,2,10$ and $\infty$ (constant duration). In (a) $\E{I}=1$ and in (b) $T_g^h=1$ from Equation \eqref{TghsSIR} are kept fixed. Note the opposite monotonic relationships between $r^h$ and $R_0^h$ when $R_0^h>1$ in the two cases. Second row: relationship between $r$ and $R_0$ on a regular network with degree $D\equiv 3$ for sSIR infectivity profiles with gamma distributed duration of the infectious period $I$ having shape parameter $\alpha = 0.5,1,2,10$ and $\infty$ (constant duration). In (c) $\E{I}=1$ and in (d) $T_g=1$ is kept fixed. Note how, as $R_0\to 2$, $r$ diverges in (c) but not in (d); the approximation $r_{\text{app}}$ does not diverge in (c) and even shrinks in (d).}
\label{fig:ModelComparison}
\end{figure}

\subsection{Estimating $R_0$ from $r$}
\label{sec:EstimatingR0fromr}
Above we have explored how the real-time growth rate and the basic reproduction number change as functions of the model parameters. However, $r$ is often one of the few readily available quantities that can be measured from attainable data. Therefore, it is interesting to explore how our estimates of $R_0$ based on observed values of $r$ are affected by the model structure. More specifically, assume we observe an exponentially growing epidemic and that we somehow obtain a measurement of the real-time growth rate, which we denote hereafter by $\hat{r}$. Assume also that the social contact structure is properly represented by a locally tree-like network (i.e.~our network model represents `the truth', which may admittedly not be realistic). We may then ask: how is ability to estimate $R_0$ affected by the potential lack of information about such contact structure? And how is it further reduced if in addition we also have no information about the infectivity profile except for the generation time?  Figure \ref{fig:EstimatingR0fromr} explores exactly these issues, for a regular network with degree 3. 

\begin{figure}
\centering
\includegraphics[width = 0.5\textwidth]{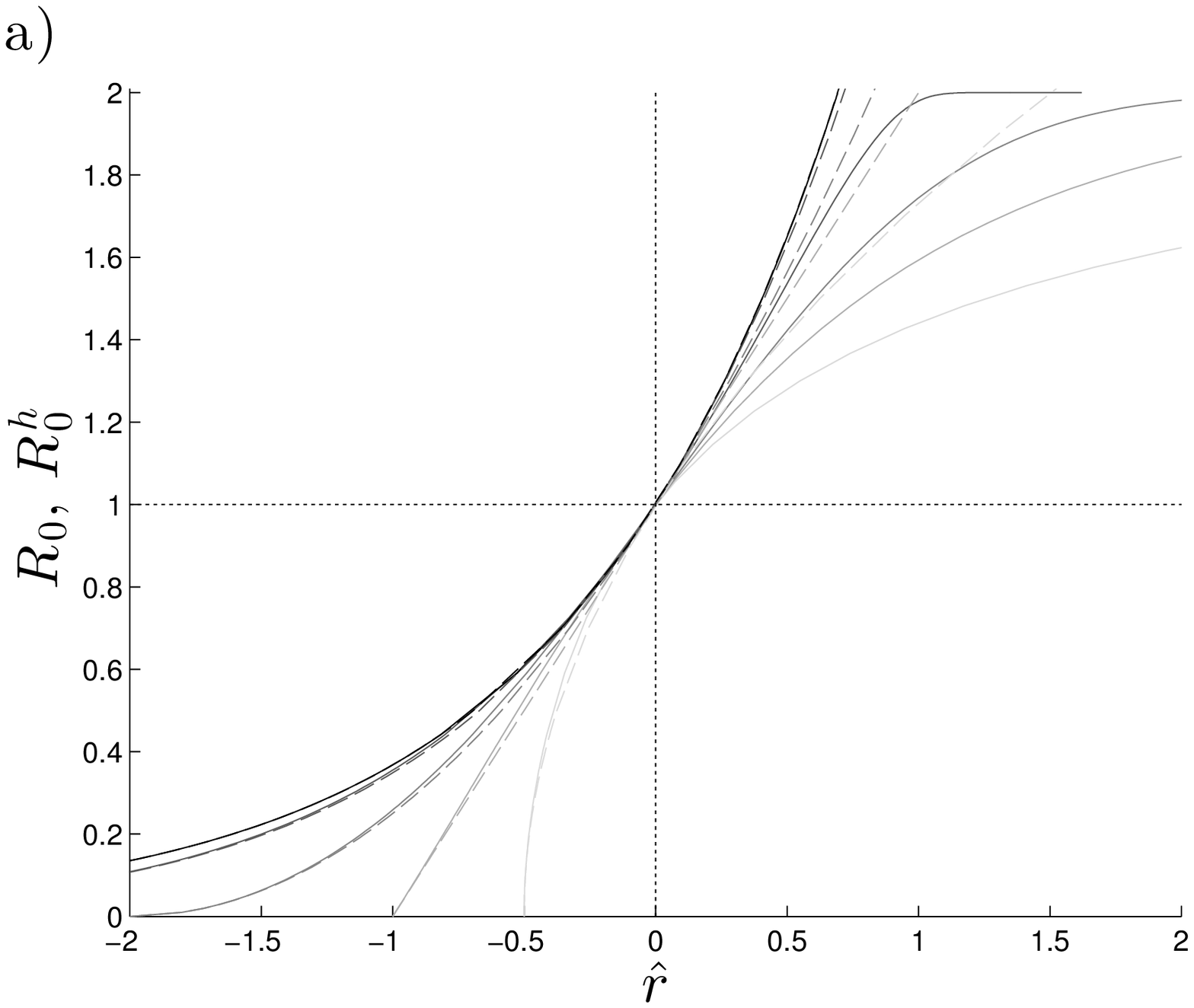}\includegraphics[width = 0.5\textwidth]{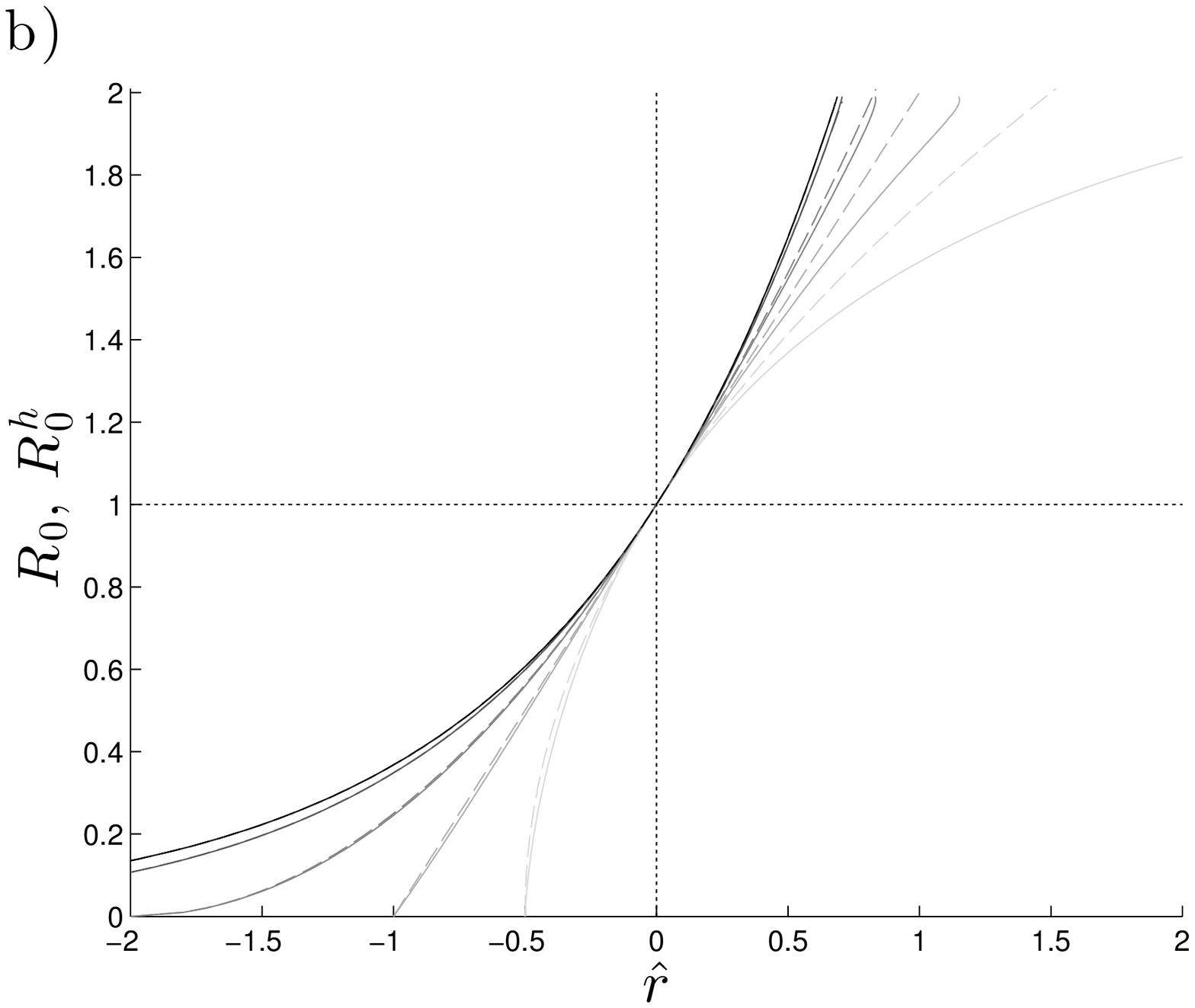}
\includegraphics[width = 0.5\textwidth]{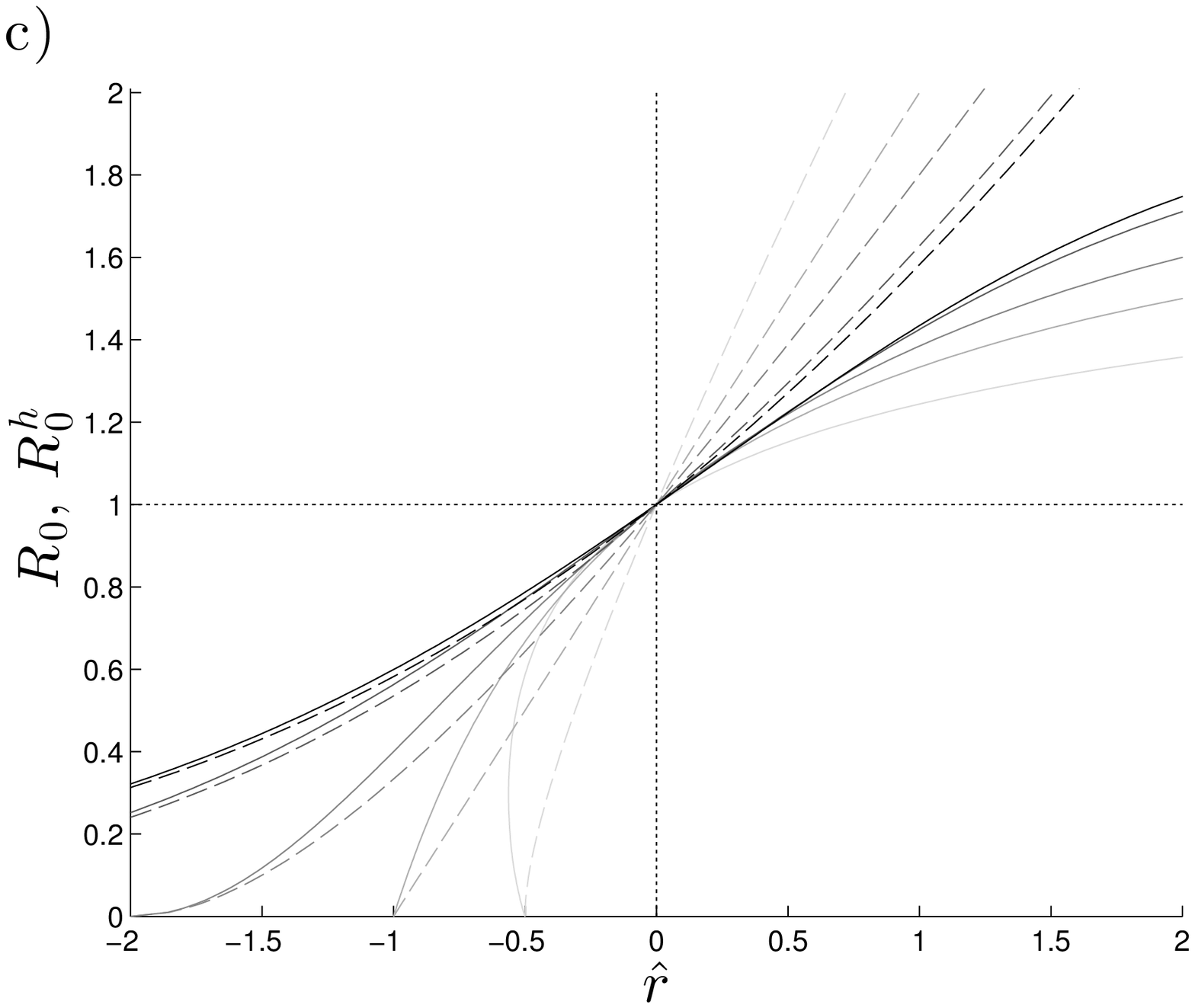}\includegraphics[width = 0.5\textwidth]{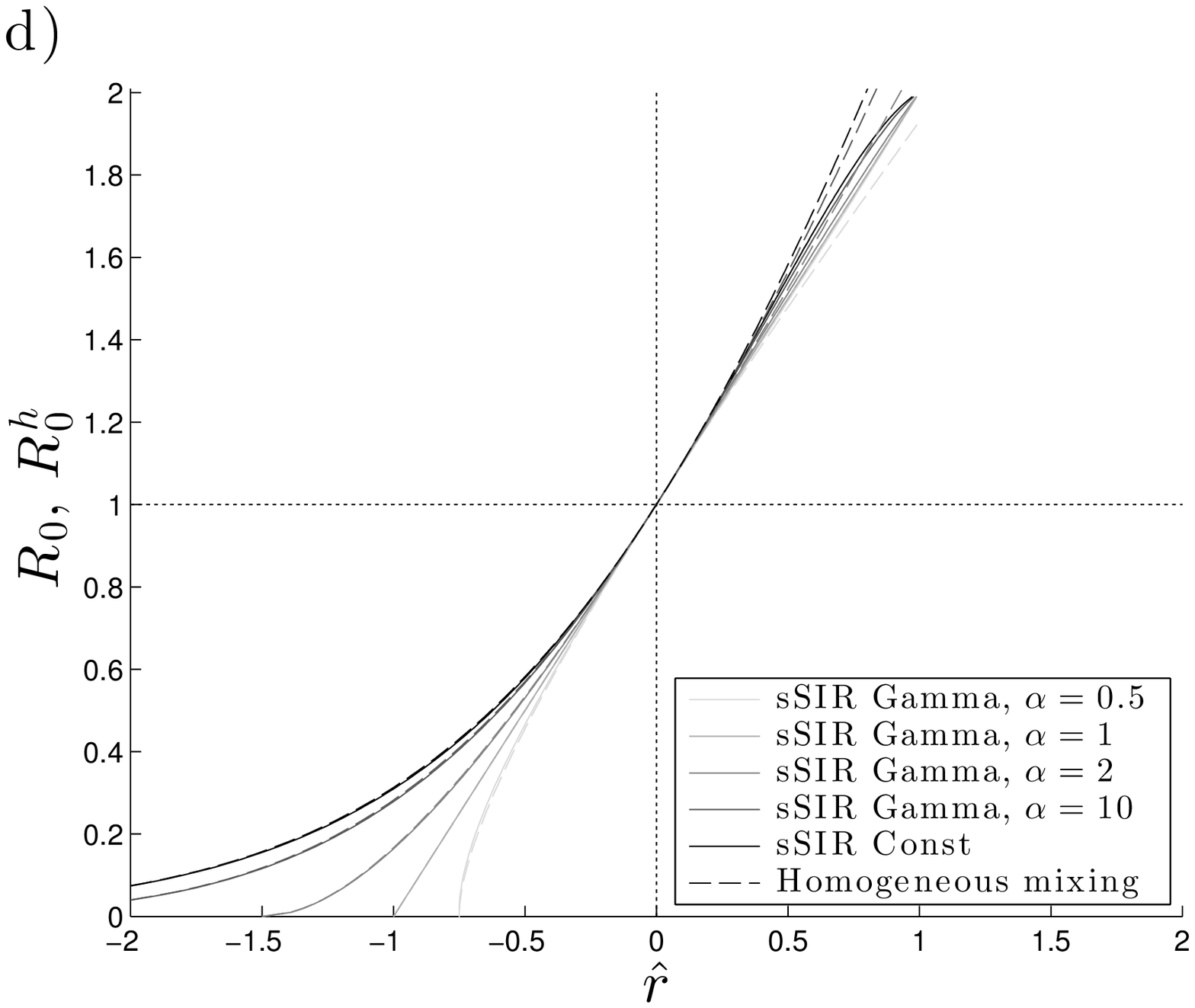}
\caption{$R_0$ as a function of $r$ on a regular network with degree $D\equiv 3$.
First row: the TVI model is assumed, having gamma shaped infectivity profiles with
shape parameters $\alpha=0.5,1,2,10$ and $\infty$. In (a) $T_g^h$ is kept fixed and in (b) we impose $T_g^h = \hat{T}_g$ (see main text). Second row: the sSIR model
is assumed, having gamma distributed duration of the infectious period $I$ with
shape parameter $\alpha = 0.5,1,2,10$ and $\infty$. In (c) $\E{I}$ is kept fixed and in (d) we impose $T_g^h=\hat{T}_g$ (see main text). The dashed lines represent the
homogeneously mixing model limit for each infectivity profile, where $r=r^h$
and $R_0=R_0^h$. The black dotted lines show where $r=0$ (vertical) and $R_0=1$
(horizontal). 
}
\label{fig:EstimatingR0fromr}
\end{figure}

In the top and bottom rows of Figure \ref{fig:EstimatingR0fromr} we consider the TVI and sSIR models, respectively. In the left column we assume $\hat{r}$ is observed and information about the infectivity profile is available by accurate observation of infected cases. Denote by $R_0(r)$ the relationship between $r$ and $R_0$ on the network and by $R_0^h(r^h)$ the relationship in a homogeneously mixing population. Then Figures \ref{fig:EstimatingR0fromr}a and \ref{fig:EstimatingR0fromr}c show the `true' $R_0(\hat{r})$ (i.e.~on the network) associated to the observed exponential growth and the $R_0^h(\hat{r})$ that would be estimated from $r^h = \hat{r}$ in a homogeneously mixing population with the same infectivity profile: in \ref{fig:EstimatingR0fromr}a the TVI profiles share the same $T_g^h$, while in \ref{fig:EstimatingR0fromr}c the sSIR profiles all have the same $\E{I}$. In the right column we take a different approach: in both cases we know the type of infectivity profile and its variability (TVI with $\alpha$ parameter for the shape in \ref{fig:EstimatingR0fromr}b and sSIR with $\alpha$ parameter for the distribution of $I$ in \ref{fig:EstimatingR0fromr}d), and in addition we assume we can measure the generation time intervals, which we denote by $\hat{T}_g$, by observing the times at which new cases arise as the epidemic unfolds. Let us denote by $R_0(r,T_g)$ the relationship between $R_0$, real-time growth rate and generation time on the network and by $R_0^h(r^h,T_g^h)$ the corresponding relationship in a homogeneously mixing population. In Figures \ref{fig:EstimatingR0fromr}b and \ref{fig:EstimatingR0fromr}d we then first plot the `true' $R_0(\hat{r}, \hat{T}_g)$ associated to the observed exponential growth and generation time, and we compare it with the $R_0^h(\hat{r},\hat{T}_g)$ that would be estimated in a homogeneously mixing population where cases would grow exponentially at a rate $r^h=\hat{r}$ with generation time $T_g^h=\hat{T_g}$ (note that here we are fixing $T_g^h$ also for the sSIR model, rather than $\E{I}$). Although our methodology is also valid when $r<0$, the setup in this section is meaningful only when a large epidemic is being observed. Therefore, we restrict our attention to the case of $r>0$. 

The first observation from Figure \ref{fig:EstimatingR0fromr} is that, at fixed infectivity profile (and for $r>0$), the estimates of $R_0$ obtained by ignoring the social structure are always conservative, i.e.~$R_0^h$ is always larger than $R_0$. This result is in line with the findings of \citet{Kenah11} and \citet{Bal+14}. The latter study, in particular, is more general as it considers additional social structures, namely multitype and households models. However, in the present work we examine more extensively the impact of different infectivity profiles.

Assume now that information about the full infectivity profile is not available: then from Figure \ref{fig:EstimatingR0fromr}a we know that, if we assume a TVI model and we only know $r$ and $T_g^h$, then ignoring the social structure and assuming a Reed-Frost model with latent period of length $T_g^h$ leads to the most conservative estimate about $R_0$. The discrepancy between $R_0^h$ and $R_0$ can potentially be quite large, especially for large $r$. However, in realistic settings the value of $\dbmd$ can be significantly larger than the value $\dbmd=2$ chosen here to accentuate network effects, leading to a much smaller discrepancy. On the other hand, if the sSIR model is assumed and only $r$ and $\E{I}$ are known, no obvious upper bound is seen in Figure \ref{fig:EstimatingR0fromr}c, as overdispersed distributions for $I$ can make the difference between $R_0^h$ and $R_0$ arbitrarily large. However, as shown in Figure \ref{fig:ModelComparison} (top line), if $T_g^h$ is kept fixed instead of $\E{I}$  the monotonic relationship in $r$ is reversed and the constant duration of the infectious period would then lead to the most conservative estimates for $R_0^h$. Furthermore, $R_0^h$ is allowed to grow unbounded because of the lack of a latent period: if such a period were added, then a combination of a random latent period followed by an infectious period concentrated in a single point in time, would lead to the most conservative estimate of $R_0^h$ (as is the case for the Reed-Frost model in Figure \ref{fig:EstimatingR0fromr}a).

Finally, if we were to decide which observation scheme to adopt between measuring $T_g^h$ (or $\E{I}$) by carefully studying the shedding profile of an infective or measuring $\hat{T}_g$ by observing the time interval of cases during an unfolding epidemic, then the latter would lead to estimates of $R_0^h(\hat{r},\hat{T}_g)$ obtained by ignoring the social structure altogether that are much closer to the ``true'' $R_0(\hat{r},\hat{T}_g)$. This is true both when the full details of the infectivity profile are known and when no information about it is available except for the estimated $\hat{T}_g$ (in particular in Figure \ref{fig:EstimatingR0fromr}d). However, we do recognise that measuring $\hat{T}_g$ from data might not be a straightforward task (even ignoring realistic logistic constraints), because of potential observation biases (see \citealp{ScaliaTomba:2010}, for a thorough discussion of such biases).

\subsubsection{Reproduction numbers and pair approximation models}

It is worth mentioning that, in addition to $R_0$, various other reproduction numbers have been defined in the literature for epidemics spreading in population with a social structure. \citet{Gold:2009} and \citet{BalPelTra14} analysed the properties of numerous such quantities in the
context of epidemic models with two levels of mixing, while
\citet{House:2011} used pair approximation methods to analyse four thresholds
for Markovian epidemic dynamics on regular networks (i.e.~where each individual has
the same number of neighbours) with significant levels of clustering, i.e.~in the presence of short loops like those in population $P_3$ in Figure \ref{threepops}. While some of these reproduction numbers are beyond the scope of our current paper, we note that one is particularly relevant for a locally tree-like network structure ($P_2$ in Figure \ref{threepops}): the exponential-growth associated reproduction number, called $R_r$ by \citet{Gold:2009} and \citet{BalPelTra14} and $r_0$ by \citet{House:2011}. This is defined as the reproduction number one would infer from knowledge of the contact interval distribution $w(t)$ and the early exponential growth rate $r$, assuming the network structure is ignored, i.e.
\begin{equation}
	R_r := \frac{1}{\int_0^{\infty} w(t) {\rm e}^{-r t} \rd t } .
	\label{Rr}
\end{equation}
Note that $R_r$ is defined as a network quantity that aims at approximating $R_0$, more than the result of a comparison between a network and a homogeneously mixing model. However, if we assume that $r = \hat{r}$ is observed, then $R_r = R^h_0(\hat{r})$ exactly as discussed in the previous section when the two models are compared at fixed contact interval distribution ($w = w^h$). Therefore the relationship between $R_0$ and $R_r$ is fully described by Figures \ref{fig:EstimatingR0fromr}a and \ref{fig:EstimatingR0fromr}c. 

The context is conceptually different when $r$ is not observed but calculated exactly from the basic parameters of the network model. This is the typical case of pair approximation models on networks. Because the focus is usually on tracking the epidemic dynamics using a small system of ODEs, these models usually assume a constant recovery rate $\gamma$ (i.e.~exponentially distributed duration of infection). The real-time growth rate $r$ is obtained by linearising the system of ODEs for the number of pairs of nodes in different states and is given by $r = \lambda(\dbmd-1) - \gamma$, in agreement with Section \ref{sec:MarkovianSIRModel}. Then, from the definition of $w(t)$ and Equation \eqref{Rr},
\begin{equation}
 \label{RrMark}
 R_r = 1+\frac{r}{\gamma} = (\dbmd-1) \frac{\lambda}{\gamma} \text{ .}
 \end{equation} 
This is a well-defined reproduction number, which shares the same threshold at 1 with $R_0$ (defined in Section \ref{sec:R0}) and whose expression might sometimes be preferable to that of $R_0$. However, it must be noted that $R_r$ does not respect the standard verbal definition of the basic reproduction number as the average number of secondary cases infected by a typical case in a large and almost fully susceptible population (i.e.~when the depletion of susceptibles can be neglected and the number of cases is growing exponentially). In particular, $R_r$ is not bounded by the average number $\dbmd$ of neighbours of a typical new infective, as $r$ has no upper bound. 

\citet{House:2011} showed, in the case of Markovian dynamics on regular networks, that $R_0$ as defined in Section \ref{sec:R0} can be
recovered from differential equation models where variables are added allowing
the numbers infected in different infectious generations to be explicitly counted. Linearisation of this extended dynamical system around the disease-free equilibrium yielded a value of $R_0$ equal to that defined in Section \ref{sec:R0}, and we conjecture, based on the results of \citet{Barbour:2013}, that a similar augmented dynamical systems analysis of non-Markovian dynamics on
configuration model networks with heterogeneous degree distributions would
yield our more general expression for $R_0$.


\section{Conclusions}
\label{sec:conclusions}

We have presented a comprehensive description of analytical and numerical computation of the real-time growth rate for Markovian and non-Markovian models of infection spread on unclustered networks, providing explicit calculations in many important special cases. 
Although the strong assumption of no clustering  represents the main limitation of the present work, exact temporal dynamics in the presence of clustering appear to be theoretically intractable and very few exact analytical results are currently available.

In the comparison between homogeneously mixing and network models, we have highlighted here how there is a strong need to be clear about how the comparison is performed, as different choices on which quantities are kept fixed in the comparison can lead to very different conclusions about the impact of networks. However, in general we can conclude that, for any chosen infectivity profile, at equal mean total infection pressure of a typical infective ($R_0^h = \dbmd \E{A}$) or at equal mean total number of new cases generated by a typical infective ($R_0^h = \dbmd \E{1-\e^A} = R_0$), the presence of a network results in faster spread compared to a homogeneously mixing model. The fundamental reason is what we called \emph{infection interval contraction}: in the presence of repeated infectious contacts towards the same susceptible, the actual infection occurs at the time of the first infectious contact, which is always stochastically smaller than the time of a randomly selected infectious contact. This also implies that, at fixed observed $r$, neglecting the network structure leads to overestimating the value of $R_0$. This has already been observed in \citet{Kenah11} and a formal proof of it can be found in \citep{Bal+14}.

We have mostly considered conditions where the effect of infection interval contractions is exacerbated by low mean degree, high infectivity and/or large variability or temporal spread of the infectious period. Some of these conditions occur in contexts such as that of sexually transmitted infections (many individuals are monogamous or nearly-monogamous and long infectivity profiles -- e.g.~for HIV -- imply many repeated contacts with the same partner) and this reinforces the standard trend of using network models in this research area. However, in most other realistic epidemic systems with a large and reasonably dispersed degree distribution and realistic infectivity profiles, the impact of infection interval contractions is often mild (if detectable at all).  This suggests that, for example in the context of airborne infections, detecting the presence of a network social structure by only looking at aggregate epidemiological measures such as $R_0$ and $r$ is unlikely, but also that the overestimate of $R_0$ from observed $r$ obtained neglecting the social structure can in general be quite accurate. The latter is fundamentally the reason why, in pair approximation models, the exponential-growth associated reproduction number $R_r$ is often numerically very similar to $R_0$. Furthermore, the same reason explains why the approximation described in Section \ref{sec:Approximater} adopted by \citet{Fraser:2007} and \citet{Pellis:2010} in the case of influenza lead to reasonably accurate approximations of $r$ or good estimates of $R_0$ and the household reproduction number from $r$.

\paragraph{Acknowledgements}
 
We gratefully acknowledge the Engineering and Physical Sciences Research Council for supporting this work and the two anonymous referees for their constructive comments, which have improved the clarity of the exposition.



\end{document}